\documentclass[runningheads]{llncs}

\usepackage{microtype}

\usepackage[normalem]{ulem}
\usepackage[utf8]{inputenc}

\usepackage{mathtools} 
\usepackage{wrapfig} 
\usepackage{subcaption}
\usepackage{placeins}

\usepackage{xcolor}
\usepackage{xspace}
\usepackage{fmtcount} 
\usepackage{xstring}
\usepackage{cite}
\usepackage{amssymb}

\usepackage{nicefrac}

\usepackage[colorlinks=true,linkcolor=blue]{hyperref}

\usepackage{amsthm}

\usepackage{tikz}
\usetikzlibrary{shapes}
\usetikzlibrary{shapes.geometric}
\usepackage{tkz-euclide}
\usetikzlibrary{calc,intersections}

\tikzset{
	square/.style={
	    inner sep=0mm,
      	outer sep=0mm,
      	rotate=0,
    	draw=none,
    	fill=white
	},
	aliveSquare/.style={
        square,
        minimum size=\innerSize mm,
        draw=none,
        fill=white
    },
	burningSquare/.style={
        square,
        minimum size=\innerSize mm,
        draw=none,
        fill=red!50
    },
  	burntSquare/.style={
        square,
        minimum size=\innerSize mm,
        draw=none,
        fill=black!20
    },
	hex/.style={
		regular polygon,
		regular polygon sides=6,
      	minimum size=\size mm,
	    inner sep=0mm,
      	outer sep=0mm,
      	shape border rotate=30,
      	rotate=0,
      	text centered,
    	draw
	},
	aliveHex/.style={
        hex,
        minimum size=\innerSize mm,
        draw=none,
        fill=white
    },
	burningHex/.style={
        hex,
        minimum size=\innerSize mm,
        draw=none,
        fill=red!50
    },
  	burntHex/.style={
        hex,
        minimum size=\innerSize mm,
        draw=none,
        fill=black!20
    },
    villageHex/.style={
        hex,
        minimum size=\innerSize mm,
        draw=none,
        fill=blue!70
    }
}

\newcommand{\terrain}[4][0.75]{%
\pgfmathtruncatemacro\height{\number\numexpr #2\relax}
\pgfmathtruncatemacro\width{\number\numexpr #3\relax}
\pgfmathsetmacro\stretch{#1}
\pgfmathsetmacro\size{\stretch*10}
\pgfmathsetmacro\innerSize{\size -1}

\begin{scope}
\foreach \row in {1,...,\height}{ 
	\ifodd\row 
    	\foreach \col in {1,...,\width}{
			\node[hex] (r\row;c\col) at ({(\col+1/2)*sin(60)*\stretch},{(-\row/2-\row/4)*\stretch}) {#4};
		}        
	\else
		\foreach \col in {1,...,\width}{
			\node[hex] (r\row;c\col) at ({\col*sin(60)*\stretch},{(-\row/2-\row/4)*\stretch}) {#4};
		}
    \fi
}
\end{scope}
}

\newcommand{\grid}[4]{%
\pgfmathsetmacro\height{\number\numexpr #1\relax}
\pgfmathtruncatemacro\width{\number\numexpr #2\relax}
\pgfmathsetmacro\stretch{.5}
\pgfmathsetmacro\innerSize{\stretch*10-1}

\begin{scope}
\draw[step=\stretch cm,] (0,0) grid (\stretch*\width,\stretch*\height);
\foreach \row in {1,...,\height}{ 
   	\foreach \col in {1,...,\width}{
        \node[#4] (r\row;c\col) at (\stretch*\col-\stretch/2,\stretch*\row-\stretch/2) {#3};
    }
}
\end{scope}
}

\newcommand{\XSays}[3]{%
}

\newcommand{\DK}[1]{{\XSays{DK}{olive}{#1}}}

\newcommand{\RK}[1]{{\XSays{RK}{blue}{#1}}}

\newcommand{\JRS}[1]{{\XSays{JRS}{green}{#1}}}

\usepackage[disable]{todonotes}
\usepackage[vlined,ruled,noend]{algorithm2e}

\DontPrintSemicolon
\SetAlCapHSkip{0em}
\SetFuncSty{textsc}
\SetProcNameSty{textsc}
\SetArgSty{textrm}
\SetKwInOut{Input}{Input}\SetKwInOut{Output}{Output}

\newcommand{\bigO}[1]{\ensuremath{O(#1)}\xspace}
\newcommand{\abs}[1]{\ensuremath{\left| #1 \right|}\xspace}

\newcommand{\ie}{i.\,e.\xspace}

\newcommand{\eg}{e.\,g.\xspace}

\newcommand{\ignitionCounter}[1][]{\ifthenelse{\equal{#1}{}}{\ensuremath{x}\xspace}{\ensuremath{x(#1)}\xspace}}
\newcommand{\burningTime}[1][]{\ifthenelse{\equal{#1}{}}{\ensuremath{y}\xspace}{\ensuremath{y(#1)}\xspace}}
\newcommand{\alive}{alive\xspace}
\newcommand{\burning}{burning\xspace}
\newcommand{\dead}{dead\xspace}

\newcommand{\rectangle}{\ensuremath{\mathcal{R}}\xspace}
\newcommand{\wall}{\ensuremath{\pi}\xspace}
\newcommand{\twist}{winding number\xspace}
\newcommand{\sourceCells}{\ensuremath{\mathcal{F}}\xspace}
\newcommand{\targetCells}{\ensuremath{\mathcal{T}}\xspace}

\newcommand{\predictedIgnitionTime}[1][]{\ifthenelse{\equal{#1}{}}{\ensuremath{t_{\text{pi}}}\xspace}{\ensuremath{t_{\text{pi}}( #1 )}\xspace}}
\newcommand{\ignitionTime}[1][]{\ifthenelse{\equal{#1}{}}{\ensuremath{t_{\text{i}}}\xspace}{\ensuremath{t_{\text{i}}( #1 )}\xspace}}
\newcommand{\queue}{\ensuremath{\mathcal{Q}}\xspace}
\newcommand{\ExtractMin}[1]{\ensuremath{\textsc{Extract-Min}(#1)}\xspace}
\newcommand{\LocalCost}[1]{\ensuremath{\textsc{Local-Cost}(#1)}\xspace}

\newcommand{\Update}[1]{\ensuremath{\textsc{Update}(#1)}\xspace}
\newcommand{\PredictIgnition}[1][]{\ifthenelse{\equal{#1}{}}{\ensuremath{\textsc{Predict-Ignition-Time}}\xspace}{\ensuremath{\textsc{Predict-Ignition-Time}(#1)}\xspace}}

\begin{document}

\title{A New Model in Firefighting Theory\thanks{This work has been supported in part by DFG grant Kl 655/19 as part of a DACH project and by NSERC under grant no. RGPIN-2016-06253.}}
%
%
\author{Rolf Klein\inst{1}
	\and David Kübel\inst{1}
	\and Elmar Langetepe\inst{1} 
	\and Jörg-Rüdiger Sack\inst{2} 
	\and Barbara Schwarzwald\inst{1}
}
\authorrunning{R. Klein, D. Kübel, E. Langetepe, J.-R. Sack and B. Schwarzwald}
%
\institute{
Department of Computer Science, Universität Bonn, 53115 Bonn, Germany \\
\email{\{rklein,dkuebel,schwarzwald\}@uni-bonn.de}\\
\email{elmar.langetepe@cs.uni-bonn.de}\\
\and
School of Computer Science, Carleton University, Ottawa, ON K1S 5B6, Canada\\
\email{sack@scs.carleton.ca}}

\maketitle              
\begin{abstract}
Continuous and discrete models \cite{article:bressan2007,article:fomin2016firefighter} for firefighting problems are well-studied in Theoretical Computer Science. 
We introduce a new, discrete, and more general framework based on a hexagonal cell graph to study firefighting problems in varied terrains.
We present three different firefighting problems in the context of this model;
for two of which, we provide efficient polynomial time algorithms and for the third, we show NP-completeness.
We also discuss possible extensions of the model and their implications on the computational complexity.

\end{abstract}
\section{Introduction and Model Definition}
\label{section:introduction}
Fighting multiple wildfires simultaneously or predicting their propagation involves many parameters one can neither foresee nor control.
For the study of problems in this context, several models have been suggested and investigated in different communities.

In Theoretical Computer Science or Mathematics, models have been investigated, where fire spreads in the Euclidean plane or along edges of a graph; see \eg \cite{article:bressan2007,article:fomin2016firefighter}.
Research in these models usually focuses on proving tight lower and upper bounds on what can be achieved with limited resources:
In continuous models, researchers have been analysing the building speed of barriers which slow down or even stop the fire's expansion;
in discrete models, the number of firefighters available to block/contain/extinguish the fire has been considered.
Tight bounds are only available for simple cases in these models, \eg \cite{article:kim2019geometric}.
\DK{I still don't like this formulation}
For a survey, we refer to \cite{survey:finbow2009}.

In other communities, models have been developed to predict a fire's propagation in a given terrain.
To make the forecast as realistic as possible, some models incorporate thermodynamic or chemical parameters as well as weather conditions including wind speed and direction.
Some of the models are capable to distinguish between fires at different heights such as ground fires and crown fires.
For a survey on theoretical and (semi-) empirical models, see \cite{survey:pastor2003mathematical}.
\todo{Mention drawbacks:
most of these models assume that each point of fire propagates independently of its neighbours.
reduce the fire front to a single line, whereas in practise a whole combustion zone should be taken into account.
}

We introduce a new model with the aim to develop a simple, theoretical framework for fire propagation forecast in large varied terrains and prove some initial results.

\begin{definition}[basic hexagonal model]
Given a partition of the plane into hexagonal cells.
The state of cell $c$ at time $t$ is given by two non-negative integers, \ignitionCounter[c,t] and \burningTime[c,t].
Cell $c$ is called \emph{burning} at time $t$ if $\ignitionCounter[c,t]=0$ and $\burningTime[c,t]>0$ hold; \emph{alive} if $\ignitionCounter[c,t]>0$ and $\burningTime[c,t]>0$; or \emph{dead} if $\burningTime[c,t]=0$ holds.
At the transition from time $t$ to $t+1$, the state of cell $c$ changes as follows:
\begin{itemize}
\item If $c$ is alive at time $t$, then $\ignitionCounter[c, t+1] := \max\lbrace \ignitionCounter[c, t] - b, 0 \rbrace$, where $b$ denotes the number of direct neighbours of $c$ burning at time $t$.
\item If $c$ is burning at time $t$, then $\burningTime[c,t+1] := \burningTime[c,t] -1$.
\end{itemize}
\end{definition}

Intuitively, \ignitionCounter and \burningTime describe the (diminished) resistance against ignition and the (remaining) fuel of an individual cell at time $t$, respectively.
Choosing suitable values for the cells, one can model natural properties of a given terrain: different types of ground and fuel; natural obstacles such as mountains or rivers.
A cell of dry grassland might get small integers for both values such that it catches fire easily and burns down quickly.
In contrast, we expect both values to be comparatively high for a moist forest such that the forest keeps burning for quite a while, once it caught fire.
\autoref{figure:exampleSpread} shows an example how a fire expands over time from a single source in a small lattice.

By this definition, our basic hexagonal model is a cellular automaton \cite{book:toffoli1987cellular}, whose cells can have state sets of different cardinality.
We observe that in the basic hexagonal model a dead cell can never become alive or burning again.
This is a major difference to cellular automata like \emph{Conway's Game of Life} \cite{article:gardner1970mathematical} or \emph{Wolfram’s model} \cite{article:wolfram1983statistical}.
Another difference is that cells can die from overpopulation in Conway's Game of Life, for which there is no equivalent rule in our model.%
%
%
\RK{
Wir haben vergessen zu sagen, ob und wie unser Modell sich von gängigen neuronalen Netzen unterscheidet.
Ein Punkt ist sicher die regelmäßige Struktur und die unveränderlichen Übergangsfunktionen.
Aber ist unser Modell am Ende ein Spezialfall?
}\DK{
I do not think we need to compare our model to (artificial) Neural Networks:\\
Both models seem similar in the sense that a cell in our model can be seen as some sort of neuron: It stimulates neighbouring cells after a predefined threshold and is stimulated by neighbouring cells itself.
However, the mutual interaction of cells in our model is very different from the predefined interaction between neurons in neural networks that is usually described by a non-planar, directed graph.
}%
\JRS{%
It is not clear how to apply it directly. There might be a way but this could be research (see below)\\
An input of raster data is well-suited to the tensors on which the neural networks typically operate. If we assume fixed dimensions for the raster data, it can be provided as input to a network in the same way as an image. Both are two-dimensional grids of intensity values. Furthermore, multiple layers of raster data can easily be combined in the input since the networks typically accept images with multiple channels (e.g., RGB).\\
Providing a useful output on a typical classification network may be difficult since they are usually used for determining a single class. There may be certain tricks that could be applied here or some form of multi-label classification that could be used. Instead, of working with a classification network, my suggestion is to try a generative network intended for style transfer. These networks take an input image and produce a new version as output which has a particular style applied to it. The idea is to consider the "style" to be the evolution of the raster data after a fixed time interval. In this way, the network can attempt to predict spread/propagation. Furthermore, if you take the output and then provide it as new input, you can continue to predict farther into the future.\\
As to whether  anything like this has been done in existing research, I'm not sure. I haven't tried looking into it and I'm not entirely sure what the best search terms would be. The word "propagation" will give results on propagation of data through the network so that won't be useful.
Some popular style transfer models are CycleGAN and DiscoGAN. There may be better models that have come since these ones. One of the challenges may be in finding large enough data sets to train the network. There is also the concern of finding appropriate training "pairs". To learn
how to apply a style to an image, the network ideally needs to know what the image looks like with and without the style in its training data. In many cases, such pairs are not available to use for training. I think the papers talk about this issue and have ways of dealing with it.
}%
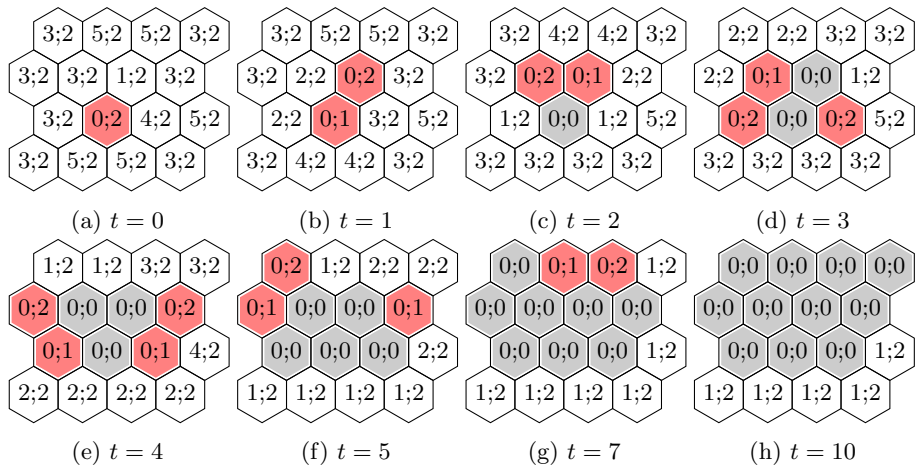
\begin{figure}
\centering
	\begin{subfigure}[b]{0.24\textwidth}
		\centering
		\begin{tikzpicture}
		\terrain{4}{4}{3;2};
		\node [burningHex] at (r3;c2) {0;2};
		\node [aliveHex] at (r1;c2) {5;2};
		\node [aliveHex] at (r1;c3) {5;2};
		\node [aliveHex] at (r2;c3) {1;2};
		\node [aliveHex] at (r3;c3) {4;2};
       \node [aliveHex] at (r3;c4) {5;2};
		\node [aliveHex] at (r4;c2) {5;2};
		\node [aliveHex] at (r4;c3) {5;2};
		\end{tikzpicture}
		\caption{$t=0$}
	\end{subfigure}
	\begin{subfigure}[b]{0.24\textwidth}
		\centering
		\begin{tikzpicture}
		\terrain{4}{4}{3;2};
		\node [burningHex] at (r2;c3) {0;2};
		\node [burningHex] at (r3;c2) {0;1};
		\node [aliveHex] at (r1;c2) {5;2};
		\node [aliveHex] at (r1;c3) {5;2};
		\node [aliveHex] at (r2;c2) {2;2};
		\node [aliveHex] at (r3;c1) {2;2};
       \node [aliveHex] at (r3;c4) {5;2};
		\node [aliveHex] at (r4;c2) {4;2};
		\node [aliveHex] at (r4;c3) {4;2};
		\end{tikzpicture}
		\caption{$t=1$}
	\end{subfigure}
	\begin{subfigure}[b]{0.24\textwidth}
		\centering
		\begin{tikzpicture}
		\terrain{4}{4}{3;2};
		\node [burntHex] at (r3;c2) {0;0};
		\node [burningHex] at (r2;c3) {0;1};
		\node [burningHex] at (r2;c2) {0;2};
		\node [aliveHex] at (r1;c2) {4;2};
		\node [aliveHex] at (r1;c3) {4;2};
		\node [aliveHex] at (r2;c4) {2;2};
		\node [aliveHex] at (r3;c1) {1;2};
		\node [aliveHex] at (r3;c3) {1;2};
       \node [aliveHex] at (r3;c4) {5;2};
    \end{tikzpicture}
		\caption{$t=2$}
	\end{subfigure}
	\begin{subfigure}[b]{0.24\textwidth}
		\centering
		\begin{tikzpicture}
		\terrain{4}{4}{3;2};
		\node [burntHex] at (r3;c2) {0;0};
		\node [burntHex] at (r2;c3) {0;0};
		\node [burningHex] at (r2;c2) {0;1};
		\node [burningHex] at (r3;c1) {0;2};
		\node [burningHex] at (r3;c3) {0;2};
		\node [aliveHex] at (r1;c1) {2;2};
		\node [aliveHex] at (r1;c2) {2;2};
		\node [aliveHex] at (r2;c1) {2;2};
		\node [aliveHex] at (r2;c4) {1;2};
       \node [aliveHex] at (r3;c4) {5;2};
		\end{tikzpicture}
		\caption{$t=3$}
	\end{subfigure}
	\begin{subfigure}[b]{0.24\textwidth}
		\centering
		\begin{tikzpicture}
		\terrain{4}{4}{2;2};
		\node [burntHex] at (r2;c2) {0;0};
		\node [burntHex] at (r2;c3) {0;0};
		\node [burntHex] at (r3;c2) {0;0};
		\node [burningHex] at (r2;c1) {0;2};
		\node [burningHex] at (r2;c4) {0;2};
		\node [burningHex] at (r3;c1) {0;1};
		\node [burningHex] at (r3;c3) {0;1};
		\node [aliveHex] at (r1;c1) {1;2};
		\node [aliveHex] at (r1;c2) {1;2};
		\node [aliveHex] at (r1;c3) {3;2};
        \node [aliveHex] at (r1;c4) {3;2};
		\node [aliveHex] at (r3;c4) {4;2};
		\end{tikzpicture}
		\caption{$t=4$}
	\end{subfigure}
	\begin{subfigure}[b]{0.24\textwidth}
		\centering
		\begin{tikzpicture}
		\terrain{4}{4}{1;2};
		\node [burntHex] at (r2;c2) {0;0};
		\node [burntHex] at (r2;c3) {0;0};
		\node [burntHex] at (r3;c1) {0;0};
		\node [burntHex] at (r3;c2) {0;0};
		\node [burntHex] at (r3;c3) {0;0};
		\node [burningHex] at (r1;c1) {0;2};
		\node [burningHex] at (r2;c1) {0;1};
		\node [burningHex] at (r2;c4) {0;1};
		\node [aliveHex] at (r1;c3) {2;2};
		\node [aliveHex] at (r1;c4) {2;2};
		\node [aliveHex] at (r3;c4) {2;2};
		\end{tikzpicture}
		\caption{$t=5$}
	\end{subfigure}
	\begin{subfigure}[b]{0.24\textwidth}
		\centering
		\begin{tikzpicture}
		\terrain{4}{4}{1;2};
		\node [burntHex] at (r1;c1) {0;0};
		\node [burntHex] at (r2;c1) {0;0};
    	\node [burntHex] at (r2;c2) {0;0};
		\node [burntHex] at (r2;c3) {0;0};
		\node [burntHex] at (r2;c4) {0;0};
		\node [burntHex] at (r3;c1) {0;0};
    	\node [burntHex] at (r3;c2) {0;0};
		\node [burntHex] at (r3;c3) {0;0};
		\node [burningHex] at (r1;c2) {0;1};
		\node [burningHex] at (r1;c3) {0;2};
		\end{tikzpicture}
		\caption{$t=7$}
	\end{subfigure}
	\begin{subfigure}[b]{0.24\textwidth}
	    \centering
    	\begin{tikzpicture}
		\terrain{4}{4}{1;2};
		\node [burntHex] at (r1;c1) {0;0};
    	\node [burntHex] at (r1;c2) {0;0};
		\node [burntHex] at (r1;c3) {0;0};
		\node [burntHex] at (r1;c4) {0;0};
		\node [burntHex] at (r2;c1) {0;0};
    	\node [burntHex] at (r2;c2) {0;0};
		\node [burntHex] at (r2;c3) {0;0};
		\node [burntHex] at (r2;c4) {0;0};
		\node [burntHex] at (r3;c1) {0;0};
    	\node [burntHex] at (r3;c2) {0;0};
		\node [burntHex] at (r3;c3) {0;0};
		\end{tikzpicture}
		\caption{$t=10$}
	\end{subfigure}
\caption{
Fire spreading in a hexagonal lattice.
The \ignitionCounter- and \burningTime-values are given in the cells as pairs of the form \ignitionCounter;\burningTime at time $t$.
The state of the cells is indicated by colours: 
Burning cells are red, alive cells are white, and dead cells are grey.
%
At time $t=10$, the propagation stops and several living cells remain.}
\label{figure:exampleSpread}
\end{figure}
\paragraph{Organization of the paper.}
The rest of this paper is organized as follows.
\autoref{section:problems} discusses three problems in context of the basic hexagonal model and their algorithmic solutions.
Variants of the basic model will be discussed in \autoref{section:modelVariants}.
We conclude with \autoref{section:conclusion}.

\section{Results for the Basic Hexagonal Model}
\label{section:problems}
There is a multitude of interesting questions one can formulate within this model.
In this paper, we will address the following three problems. 

Suppose we restrict the hexagonal grid to a rectangular domain \rectangle consisting of $n$ cells.
Let us assume that initially all cells along the right boundary of \rectangle are on fire and all cells along the left boundary of \rectangle represent a village that must be protected from the fire.
To this end, we want to connect the upper to the lower boundary of \rectangle by a path \wall of cells that separates the village on the left and the fire approaching from the right; see \autoref{subfigure:villageProtection}.
To make \wall fire-resistant, we can fortify the cells on \wall by increasing their \ignitionCounter-values.

In the \emph{first version} of this problem, all cells on \wall will have their \ignitionCounter-values raised by the same amount $k$.
This corresponds to a fly-over by aircraft that douses each cell with the same amount of water.
We want to compute the minimum $k$ for which such a protecting path \wall exists.
In \autoref{subsection:airplane}, we present a solution that runs in time \bigO{n \log n \log Y} where $Y$ is the maximum sum of \burningTime-values of direct neighbours over all cells in \rectangle.
Our algorithm is based on a fast propagation routine that is interesting in its own right for simulation purposes.

In the \emph{second version} of the above problem, firefighters can increase \ignitionCounter-values of cells individually.
Now we are interested in finding a separating path \wall for which the sum of these \ignitionCounter-increments of cells on \wall is minimal.
Although this appears to be a shortest-path problem, we have not been able to apply a classic graph algorithm like Dijkstra for reasons that will be explained in \autoref{subsection:villageProtection}.
Our algorithm runs in time \bigO{n \sqrt{n} \log n}, provided that all cells have identical \burningTime-values and each \ignitionCounter-value is upper bounded by $2\burningTime+1$.

In the \emph{third version}, we no longer assume that cells along the right boundary of \rectangle are on fire, while cells along the left edge have to be protected. 
Instead, the cells of the village are given by a set \targetCells and the fire is allowed to start at cells of a set \sourceCells; see \autoref{subfigure:npCompleteness}.
We now ask for a subset of \sourceCells with $m$ cells that will, when put on fire, burn all cells of the village to the ground.
More precisely, we are interested in answering the following decision problem:
Are there $m$ cells in \sourceCells which, when put on fire, will eventually ignite all cells in \targetCells?
In \autoref{subsection:npCompleteness} we prove this problem to be NP-complete.

\begin{figure}[tb]
\centering
    \begin{subfigure}[b]{0.47\textwidth}
        \centering
        \begin{tikzpicture}
        \terrain[0.25]{20}{24}{};
        \foreach \row in {1,...,20}{ 
            \foreach \col in {14,...,24}{
                \node [burningHex] at (r\row;c\col) {};
            }
        }
        \foreach \row in {6,...,13}{ 
            \foreach \col in {10,...,13}{
                \node [burningHex] at (r\row;c\col) {};
            }
        }
        \foreach \row in {1,...,20}{
            \node [villageHex] at (r\row;c1) {};
        }		
        \foreach \col [count=\i] in {14,13,12,12,11,10,9,8,7,8,9,10,10,11,13,14,14,13,12,13}{
            \node [burntHex] at (r\i;c\col) {};
        }
        \node [burningHex] at (r19;c13) {};
        \node [burningHex] at (r3;c13) {};
        \node [burningHex] at (r4;c13) {};
        \node [burningHex] at (r5;c12) {};
        \node [burningHex] at (r5;c13) {};
        \node [burningHex] at (r8;c9) {};
        \node [burningHex] at (r9;c8) {};
        \node [burningHex] at (r9;c9) {};
        \node [burningHex] at (r10;c9) {};

        \node [burntHex] at (r2;c14) {};
        \node [burntHex] at (r6;c11) {};
        \node [burntHex] at (r7;c8) {};
        \node [burntHex] at (r11;c8) {};
        \node [burntHex] at (r13;c12) {};
        \node [burntHex] at (r13;c13) {};
        \node [burntHex] at (r14;c12) {};
        \node [burntHex] at (r14;c14) {};
        \node [burntHex] at (r18;c14) {};
        
        \node [aliveHex] at (r5;c20) {};
        \node [aliveHex] at (r6;c20) {};
        \node [aliveHex] at (r5;c21) {};
        \node [aliveHex] at (r6;c21) {};
        \end{tikzpicture}
        \caption{
        A path \wall (grey cells) prevents a fire (red cells) from reaching the village (blue cells) at the left boundary of \rectangle.}
        \label{subfigure:villageProtection}
	\end{subfigure}
	~
	\begin{subfigure}[b]{0.47\textwidth}
		\centering
		\begin{tikzpicture}
        \terrain[0.25]{20}{24}{};
        \node [burningHex] at ( r2;c5) {};
        \node [burningHex] at ( r3;c20) {};
        \node [burningHex] at ( r6;c12) {};
        \node [burningHex] at ( r7;c7) {};
        \node [burningHex] at (r10;c17) {};
        \node [burningHex] at (r14;c4) {};
        \node [burningHex] at (r17;c13) {};
        \node [burningHex] at (r18;c19) {};
        \node [villageHex] at ( r2;c12) {};
        \node [villageHex] at ( r4;c16) {};
        \node [villageHex] at ( r5;c14) {};
        \node [villageHex] at ( r8;c3) {};
        \node [villageHex] at ( r9;c21) {};
        \node [villageHex] at (r11;c12) {};
        \node [villageHex] at (r13;c9) {};
        \node [villageHex] at (r18;c19) {};
        \node [villageHex] at (r19;c6) {};        
        		\end{tikzpicture}
		\caption{
		Given two set of cells \sourceCells (red) and \targetCells (blue).
		Can $m$ cells out of \sourceCells burn down \emph{all} cells of \targetCells, when put on fire?}
		\label{subfigure:npCompleteness}
	\end{subfigure}
	\caption{Problem variants on a rectangular finite domain \rectangle.}
	\label{figure:problemVariants}
\end{figure}
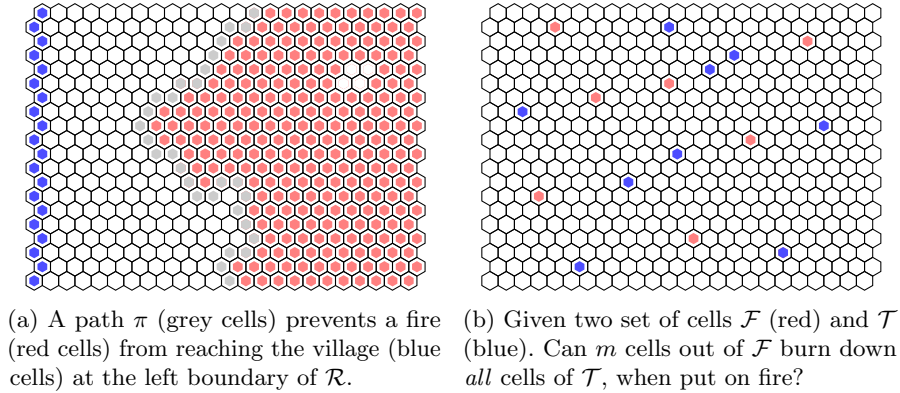
\subsection{Homogeneous fortification}
\label{subsection:airplane}
For this problem, consider a rectangular domain \rectangle of our basic hexagonal model, in which all cells on the right boundary are on fire and all cells on the left boundary represent a village that must be protected.
We call a path \wall connecting the lower to the upper boundary of \rectangle a \emph{separating path}.
We call $\wall$ a \emph{protecting path} for $k$ if increasing the 
\ignitionCounter-value of all cells along \wall by $k$ ensures that the fire never ignites a cell of the village.
The natural optimization problem is to find the minimum $k$ for which a protecting path exists.

To solve this problem, we study the corresponding decision problem:
Given \rectangle and $k$, does a protecting path \wall exist?
For $k=0$, it can be solved by simulating the fire propagation step-by-step, where all cells on the right boundary of \rectangle are initially ignited.
Consider the map at the end of the simulation:
All cells are either \alive or \dead but none \burning; all \dead cells form a connected component including the initially \burning cells on the right boundary of \rectangle.
If no protecting path for $k=0$ exists, at least one of the village cells will be \dead.
However, if a protecting path for $k=0$ exists, then some of the \alive cells form a connected component that includes all village cells on the left boundary of \rectangle, see \autoref{subfigure:villageProtection}.
The \emph{fire border} of this component which are \alive cells with a direct \dead neighbour, form a protecting path \wall for $k=0$.

This approach can be extended to solve the decision problem for larger values of $k$:
First increase the \ignitionCounter-value of all cells by $k$, then run the simulation algorithm.
If a cell of the village is \dead at the end of the simulation, no protecting path for $k$ exists.
Otherwise, consider the fire border of the connected set of \alive cells that includes the village and induces the separating path \wall.
By construction, all cells of \wall stay \alive, when their \ignitionCounter-value is increased by $k$.
This holds even if all cells right of \wall are \burning or \dead.
This also holds when the \ignitionCounter-values of all cells of $\rectangle\setminus \wall$ remain untouched:
Increasing the \ignitionCounter-value of all \burning cells is irrelevant for the survival of the village;
increasing the \ignitionCounter-value of any other \alive cell is irrelevant as well, since only \alive cells on \wall have \burning or \dead neighbours.
Therefore, \wall is a protecting path in \rectangle for $k$.

To solve the optimization problem, we combine the decision-algorithm with binary search.
It remains to give a sensible upper bound on $k$ and an efficient algorithm for the simulation of fire propagation.
Let $Y$ be the maximum sum of \burningTime-values of all direct neighbours of a cell, over all cells in the grid.
As every cell's \ignitionCounter-value can only be decreased by at most the sum of its neighbours \burningTime, we know that $0 < k\leq Y$ holds.
A brute force step-by-step simulation over time results in an algorithm with a worst-case running time of \bigO{n^2 \cdot x_{max}}, where $x_{max}$ is the maximal value of \ignitionCounter over all cells.
However, given the state of all \burning or \dead cells at a time $t$, one can determine the next cell to ignite.
This intuition gives rise to an \bigO{n \log n} Dijkstra-inspired algorithm independent of the cells values, which uses a priority queue for the retrieval of the next cell to ignite.
Details and proofs can be found in Appendix \ref{appendix:fastSimulationAlgorithm}.

\begin{theorem}
Let \rectangle be a rectangular domain, where all cells along the right boundary of \rectangle are on fire and all cells along the left boundary have to be protected.

The minimum $k$ for which a protecting path exists can be found in time \bigO{n \log n \log Y} where $Y$ denotes the maximum sum of \burningTime-values of all direct neighbours of a cell, over all cells.
\end{theorem}

\subsection{Selective fortification}
\label{subsection:villageProtection}
Similar to \autoref{subsection:airplane}, consider a rectangular domain \rectangle of our basic hexagonal model in which all cells on the right boundary are on fire and all cells on the left boundary represent a village that must be protected.
We call a path \wall connecting the lower to the upper boundary a \emph{separating path}.
The path can be fortified to protect the village by individually increasing the \ignitionCounter-value of each cell along \wall.
For a given path \wall, we call the sum of those increments the \emph{fortification cost} of \wall.
The natural optimization problem is to find a separating path \wall that minimizes the fortification cost.

To begin with, observe that there is always a separating path with minimal fortification cost such that every cell along \wall has a direct, \dead neighbour when the simulation ends:
Any cell of \wall without a direct, \dead neighbour can be excluded from \wall without increasing the fortification cost;
all cells with a direct, \dead neighbour that do not belong to \wall can be included to \wall since they do not require any fortification costs at all.
Doing this, we obtain a separating path where all cells have \dead neighbours.
Therefore, we may restrict our search to separating paths to the right of which all cells are \dead.

Moreover, this observation allows to compute the fortification costs of such a path:
For a cell $c$ of \wall, let $Y_r$ be the sum of $y$-values of all neighbours of $c$ to the right of \wall.
The fortification cost of $c$ is the minimum $k$ such that $\ignitionCounter[c,0] + k = Y_r + 1$.
The fortification cost of \wall is the sum of the fortification costs of all cells of \wall.

Finding a separating path \wall of cells is equivalent to finding a path $\wall_b$ along corners and edges of the cells:
Using the observation stated above, we may conclude that there is such a path $\wall_b$, where cells to the left belong to \wall and cells to the right are \dead.
This path $\wall_b$ lies in the graph given by the corners and edges of the cells, which we call the border graph.
Similar to the previous denotation, we call a path $\wall_b$ in the border graph \emph{separating} if it connects the upper to the lower boundary of \rectangle. 
To distinguish left from right, we replace every edge $\{v,w\}$ in the border graph by two directed edges $(v,w)$ and $(w,v)$.

\begin{wrapfigure}{l}{0.35\textwidth}
	\centering
	\includegraphics[width=0.35\textwidth]{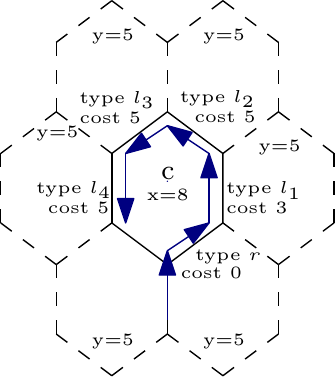}
	\caption{A local-cost example.}
	\label{fig:costExample}
\end{wrapfigure}
When transforming the optimization problem on the cells into a shortest-path problem on the border graph, it is not obvious how to assign the fortification cost of a cell to the adjacent edges;
consider the example depicted in \autoref{fig:costExample}, where $\burningTime = 5$ for all cells and cell $c$ has $\ignitionCounter = 8$.
The path $\wall_b$ in question uses five edges of $c$, whose right neighbour cells are considered as burning.
The crucial idea is to charge these edges for the fortification cost of $c$ depending on their occurrence in $\wall_b$:
The first directed edge of $\wall_b$ along $c$ gets cost $0$ because no additional fortification is necessary to protect $c$ from a single burning neighbour; 
the second edge gets cost $3$;
every further edge gets cost $5$, since the \ignitionCounter-value of $c$ has to be increased by $5$ for every additional \burning neighbour.
Unfortunately, this dynamic assignment of costs, where the cost of an edge depends on the previous edges of the path, rules out a direct solution via finding a shortest-path:
A shortest path might visit the border of a cell several times; the edges along the same cell do not necessarily have to lie on the path in direct succession.
Hence, the cost of an edge can be influenced by any previous edge in the path.
\DK{%
Note that such shortest-path problems with dynamic edge weights are often NP-hard, \eg:
If, at a vertex $v$, we can change the cost of an edge from 1 to $n$ depending on whether the number of edges used in the path that leave $v$ is non-zero, we could reduce Hamilton Path to this shortest-path problem.
}%

In general, the following two problems rule out a direct solution via shortest-path finding algorithms:
(1) A shortest path $\wall_b$ does not have to be simple and can have self-intersections, see \autoref{fig:intersectionExample};
(2) edges of $\wall_b$ along the same cell $c$ do not have to lie on $\wall_b$ in direct succession, it can leave and \emph{revisit} $c$ multiple times, see \autoref{fig:revisitExample}.
In the following, we consider the problem for the case where the \burningTime-values are identical for all cells and $0 < \ignitionCounter[c,0] \leq 2 \burningTime + 1$ holds.
This implies that each cell can always be ignited by three direct, \burning neighbours.
Based on these assumptions, we are able to prove that none of these problems occurs for a so-called shortest local-cost path.

Let $e$ be an edge along a cell $c$ of resistance $\ignitionCounter_e:=\ignitionCounter[c,0]$.
We say $e$ is of type $r$ if it comes after a right turn or is the very first edge of path $\wall_b$.
We say $e$ is of type $l_k$ if it comes after the $k^{th}$ consecutive left turn of $\wall_b$ along $c$.
Thus, $e$ is the $(k+1)^{th}$ consecutive edge along the same cell $c$ on the left-hand side of $\wall_b$.
Thus, we can define the local-cost of an edge $e$ at cell $c$ depending on $\ignitionCounter_e$ and its type:
\begin{equation*}
c(e, \text{type}) =
\begin{cases}
\max (0, y+1-\ignitionCounter_e), & \text{if type} = r\\
\min (y, 2y+1-\ignitionCounter_e), & \text{if type} = l_1\\
y, & \text{if type} \in \lbrace l_2,l_3,l_4\rbrace. \\
\end{cases}
\end{equation*}

\begin{definition}[Shortest local-cost path]
Let $s, t$ be vertices in the border graph, where all cells have identical \burningTime-values and $0<\ignitionCounter[c,0]<2\burningTime[c,0] + 1$ holds.
Then, a shortest local-cost $s$-$t$-path is a path from $s$ to $t$ of minimum local edge cost as defined above.
\end{definition}

Moreover, we call the difference between the number of $r$-edges and $l_k$-edges of a path $\wall_b$ in the border graph the \emph{\twist} of $\wall_b$.

In the following, we prove that there is a shortest local-cost path $\wall_b$ with \twist one that neither has (1) self-intersections, as shown in \autoref{fig:intersectionExample}, nor (2) revisits as shown in \autoref{fig:revisitExample}:
\autoref{lem:windingNumber} and \autoref{lem:intersection} together prove (1), while \autoref{lem:revisits} proves (2).
Thus, the assigned local-cost of $\wall_b$ are the true fortification costs of the corresponding separating path \wall.
\begin{figure}[tb]
	\centering
	\begin{subfigure}[b]{0.51\textwidth}
		\includegraphics[width=\textwidth]{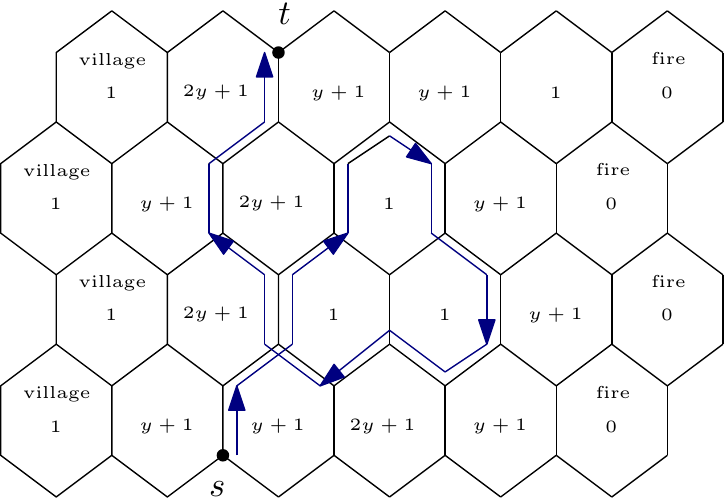}
		\caption{
        A self-intersecting, shortest local-cost \mbox{$s$-$t$-path} with cost~$0$ and \twist~$7$.
		}
		\label{fig:intersectionExample}
	\end{subfigure}
	~
	\begin{subfigure}[b]{0.43\textwidth}
		\includegraphics[width=\textwidth]{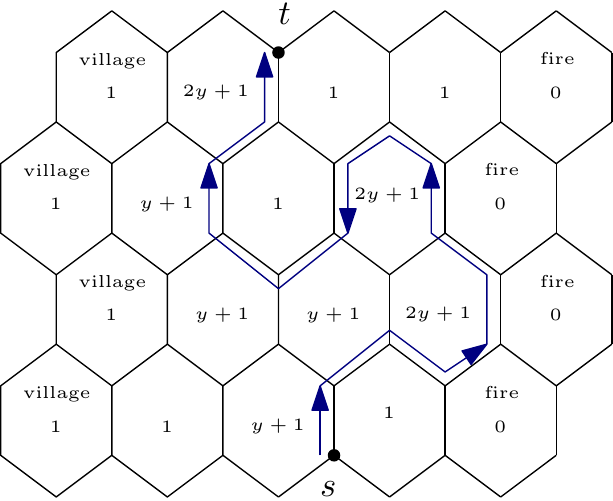}
		\caption{
		A simple $s$-$t$-path $\wall_b$ with local-cost $4y$, \twist $1$, and a revisit.
		The corresponding separating path \wall has fortification cost $5y$.
		}
		\label{fig:revisitExample}
	\end{subfigure}
	\caption{
	Two example domains which illustrate problems (1) and (2) for standard shortest-path finding algorithms.
	All cells have the same \burningTime-value and \ignitionCounter-values as denoted in the cells.
	}
	\label{fig:problemPaths}
\end{figure}

\begin{lemma}
    \label{lem:windingNumber}
	Let $s,t$ be vertices on the upper and lower boundary of the border graph of \rectangle and $\wall_b$ be an $s$-$t$ path.
	If $\wall_b$ is simple, then it has \twist~$1$.
\end{lemma}

\begin{proof}
	As all turns in the regular grid have exactly the same angle, the \twist is a measure of the turn angle of the path.
	For any right turn, the total angular turn of the path decreases by $\frac{\wall}{3}$ and every left turn increases the turn angle by $\frac{\wall}{3}$.
	A simple $s$-$t$-path has a turn angle of $0$ as first and last edge are both vertical.
	Hence, the number of left turns equals the number of right turns in $\wall_b$.
	Since the first edge is considered to be an $r$-edge, to assure that local-costs of the first cell are well defined, $\wall_b$ has \twist~$1$.
\end{proof}

Note that this does not directly solve our first problem, as it holds only for one direction: an $s$-$t$-path with \twist $1$ might still contain intersections.
The reverse holds because of our restrictions on the $\ignitionCounter$- and $\burningTime$-values.

\begin{lemma}
	\label{lem:intersection}
	Any shortest local-cost $s$-$t$-path with \twist $1$ is simple.
\end{lemma}

\begin{proof}[Proof sketch]
Assume $\wall_b$ is a shortest local-cost $s$-$t$-path with \twist $1$ and at least one intersection.
Then, we prove a contradiction by constructing an intersection-free path $\wall_b'$ with fewer costs.

Let $\wall_b$ be given by the sequence $s, v_1, v_2, \ldots, v_i, \ldots, v_{n-1}, v_n, t$ of $n+2$ vertices in the border graph.
Let $i<j$ be the smallest indices such that $v_i = v_j$ holds.
By removing all vertices between $v_i$ and $v_{j+1}$, we obtain a new path of which we show that its cost is strictly less than the cost of $\wall_b$.

While removed edges can no longer contribute to the cost of the path, removing them can change the type and hence the local cost of the edges from $v$ up to the first unaffected $r$-edge after $(v,v_{j+1})$.
A detailed proof, which shows that the local-cost of $\wall_b'$ is strictly less than the local-cost of $\wall_b$, can be found in Appendix \ref{appendix:intersectionLemma}.

All in all, repeated removal of loops results in a simple $s$-$t$-path with fewer costs than $\wall_b$ which contradicts the assumption and completes the proof.
\end{proof}

\begin{lemma}
	\label{lem:revisits}
	For any shortest local-cost $s$-$t$-path \wall with \twist $1$, there exists a simple, shortest local-cost path without cell revisits that costs not more than \wall.
\end{lemma}
\begin{proof}
	Assume $\wall_b$ is a shortest local-cost s-t-path with a \twist~1 where at least one cell on the left-hand of $\wall_b$ is revisited.
	By \autoref{lem:intersection}, $\wall_b$ is free of intersections.
	Then, we prove a contradiction by constructing a cell-revisit-free path $\wall_b'$ with local-costs no more than that of $\wall_b$.
	
	Let $e$ be the first edge on the path along a revisited cell $c$.
	Let the edges of $c$ be numbered counter clockwise from $0$ to $5$, where $e$ is edge $0$.
	Then, the situation at $e$ can be restricted to the following two cases, also illustrated in \autoref{fig:revisitCases}.
	
	\begin{enumerate}
		\item The last edge along $c$ on $\wall_b$ is edge $2$ and $1$ does not lie on $\wall_b$.
		\item The last edge along $c$ on $\wall_b$ is edge $3$ and at least one of the edges $1$ and $2$ does not lie on $\wall_b$.
	\end{enumerate}	

	Note that $\wall_b$ can neither include edge~$4$ nor $5$.
	If it included edge~$5$, $\wall_b$ would not be intersection free.
	If it included edge~$4$ but not $5$, the edge preceding $e$ on $\wall_b$ would also lie along a revisited cell and we assumed $e$ to be the first such edge on $\wall_b$.
%
%
%
%
	
	\begin{figure}[tb]
		\centering
		\includegraphics[scale=1]{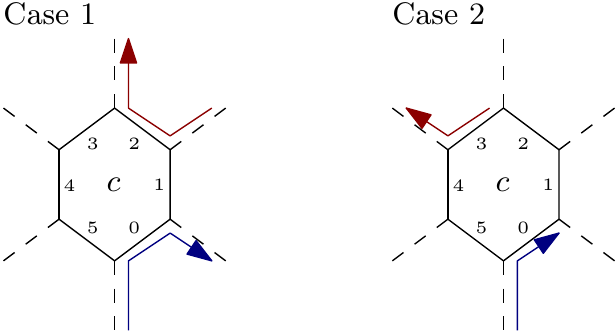}
		\caption{Path $\wall_b$ follows the blue edges upon its first visit of $c$ and the red edges on its second visit. Cases are equivalent for rotation.}
		\label{fig:revisitCases}
	\end{figure}
	
	We can construct $\wall'_b$ from $\wall_b$ by removing cell visits as follows.
	In case~$1$, we can replace all edges on $\wall_b$ after edge~$0$ and before edge~$2$ by edge~$1$.
	This adds at most cost~$2y$ ($y$ for the new edge~$1$, and another $y$ for the change of type edge~$2$ to $l_2$).
	In case~$2$, we replace all edges on $\wall_b$ between edge~$0$ and edge~$3$ including these two edges by edge~$4$ and $5$.
	This also adds at most cost~$2y$ ($y$ for edge~$4$ and $5$ each).
	However, as $\wall_b$ is free of intersection by \autoref{lem:intersection}, we know for both cases that the removed part includes enough $l_k$ edges with $k>1$ to counter the cost of adding the new edges to the path.
\end{proof}

Finally, we describe how to compute such a shortest local-cost path with a Dijkstra-inspired shortest-path algorithm tracking the \twist of the path:
We use a priority queue of tuples $(c,v,p,e,w)$, where $c$ is the minimum known local-cost of a path to a vertex $v$ via a predecessor $p$, where the last edge $(p,v)$ on the path has edge type $e$ and the whole path has \twist $w$.
Due to regularity of the lattice, we know that: the vertex degree and hence the number of possible predecessors is constant; the number of edge types is constant.
Moreover, $w$ 
is limited by the size of our grid:
A non-intersecting path with a very high \twist roughly forms a spiral, where the \twist corresponds to the number of spiralling rounds;
the maximum size of such a spiral is limited by the width $w$ and height $h$ of the grid.
Thus, it suffices to consider tuples with $|w| \leq 6 \min(w,h) = \bigO{\sqrt{n}}$.

Altogether, our priority queue contains at most \bigO{n \sqrt n} many items, which results in a runtime of \bigO{n \sqrt{n} \log n } to find a simple, shortest local-cost path from a specific starting vertex $s$ on the lower boundary to any vertex $t$ on the upper boundary.
To find the optimal separating path \wall, we have to compare shortest paths for all pairs of $s$ and $t$.
We can do this in a single run of the algorithm by initialising our priority queue with the outgoing edges of all possible $s$. 
We can terminate as soon as the minimal entry in our priority queue is a tuple where $v$ is one of the vertices along the upper boundary of our grid and $w$ is $1$.
A pseudocode description of this algorithm can be found in Appendix \ref{appendix:IndividualAlgorithm}.

\begin{theorem}
Let \rectangle be a rectangular domain, where all cells have identical \burningTime-values and $0 < \ignitionCounter[c,0] \leq 2  \burningTime[c,0] + 1$ holds.
All cells at the right boundary of \rectangle are on fire and all cells along the left boundary have to be protected.

Then, we can compute a separating path of minimum fortification cost in time \bigO{n \sqrt{n} \log n}.
\end{theorem}
For arbitrary values of \ignitionCounter and \burningTime, it is still open whether the optimization problem can be solved in polynomial time.
While the definition of local-costs can be adjusted, problems (1) and (2) remain.

\subsection{An NP-complete problem}
\label{subsection:npCompleteness}
In this section, we no longer assume that cells along the right boundary of \rectangle are on fire, while cells along the left edge have to be protected.
Instead, we consider two sets of cells \sourceCells, \targetCells, where the fire is allowed to start from cells in \sourceCells to burn cells in \targetCells.
We consider the following decision problem:
Are there $m$ cells in \sourceCells which, when put on fire, will eventually ignite all cells in \targetCells?

We prove this problem to be NP-complete by reduction from planar vertex cover.
The planar vertex cover problem is as follows:
Given a planar graph $G$, a vertex cover for $G$ is a subset of vertices that contains at least one endpoint of every edge.
This problem was proven to be NP-complete, even for planar graphs with maximum vertex degree three \cite{article:garey1977rectilinear}.%
\DK{For a planar graph there exists a planar grid embedding iff and only if the graph has a maximum degree of four.}

\newcommand{\rectiGrid}{\ensuremath{G_\square}\xspace}
Given a planar graph $G$ with $n$ vertices, we have to show how to obtain an instance of the fire expansion problem in polynomial time.
For simplicity, we assume that, instead of the basic hexagonal mode, we can use a basic grid model where the same rules apply only that a rectilinear grid is used instead of the hexagonal one.
Consider the \emph{rectilinear grid}, which is the infinite plane graph with a vertex at every positive integer coordinate and an edge between every pair of vertices at unit distance.
In a first step, we compute a planar grid embedding \rectiGrid of $G$ into the rectilinear grid such that the following holds:
Disjoint vertices of $G$ are mapped to disjoint integer coordinates; edges of $G$ are mapped to rectilinear paths in the grid such that no two paths have a point in common, except, possibly, for the endpoints.
The size of \rectiGrid is polynomial in $n$ and can be computed in polynomial time; see \cite{book:eiglsperger1999orthogonal}.
In a second step, we scale \rectiGrid by a factor of two: a vertex at coordinates $(a, b)$ is mapped to coordinates $(2a, 2b)$; edges of \rectiGrid are stretched accordingly.
This introduces \emph{buffer coordinates} between different edges in the embedding of \rectiGrid.
Finally, we place a cell at every integer coordinate spanned by \rectiGrid and set the $(\ignitionCounter, \burningTime)$-values of the cell $c$ as follows:
\begin{itemize}
\item If $c$ corresponds to a vertex of \rectiGrid (vertex-cell) set the weights to $(n, 1)$;
\item if $c$ corresponds to an edge of \rectiGrid (edge-cell) set the weights to $(1, 1)$;
\item set the weights of all remaining cells to $(0, 0)$.
\end{itemize}
Add all vertex-cells to \sourceCells and all edge-cells incident to a vertex-cell to \targetCells.
The size of the resulting instance and the construction time are polynomial in $n$.

It remains to prove that this instance has $m$ cells that ignite all cells in \targetCells iff $G$ has a vertex cover of size $m$.
On the one hand, if $G$ has a vertex cover $\mathcal{C}$ of size $m$, the corresponding vertex-cells can be chosen to put on fire.
Due to the choice of values, they will ignite all adjacent edge-cells.
Since $\mathcal{C}$ is a vertex cover, all edge-cells will burn and therefore all cells in \targetCells.
On the other hand, assume there is a subset $\mathcal{S}\subset \sourceCells$ of size $m$ that will eventually ignite all cells of \targetCells.
Due to the construction, all cells of $\mathcal{S}$ are vertex-cells, so only vertex-cells are put on fire.
Observe that due to the choice of weights, a burning edge-cell can never ignite a neighbouring vertex-cell.
Moreover, due to the buffer coordinates, any two edge-cells that belong to different edges of \rectiGrid are separated by at least one cell of weight $(0,0)$.
Consequently, every edge-cell in \targetCells must be either ignited by its direct neighbouring vertex-cell or by a fire reaching it from the direct neighbouring edge-cell, emanating from a different vertex-cell.
Since for every edge in $G$ there is an edge-cell in \targetCells, which is ignited via one of the adjacent vertex-cells of $\mathcal{S}$, the vertices of $G$ corresponding to the vertex-cells in $\mathcal{S}$ constitute a vertex cover of size $m$ for $G$. 

Certainly the problem is in NP, since the subset-certificate of cells $\mathcal{S}$ can be verified in polynomial time via standard simulation.
Thus, we obtain the following theorem.
\begin{theorem}
\label{theorem:npCompleteness}
Given an instance of the basic hexagonal model, $m\in\mathbb{N}$ and two finite sets of cells $\mathcal{F}, \mathcal{T}$.
It is NP-complete to decide whether there is a subset $\mathcal{S} \subseteq \mathcal{F}$ with $\abs{S} \leq m$ such that putting all cells of $\mathcal{S}$ on fire will eventually ignite all cells of $\mathcal{T}$.
\end{theorem}

\todo{Is our problem FPT in solution size? Planar vertex-cover is.}
Note that the hardness proof requires \targetCells to be possibly of size linear in \abs{\sourceCells}.
Restricting \targetCells to a single cell, we obtain a problem that might very well be easier to solve.
We do not know, whether this restricted problem is still NP-hard and leave this open.
Definitely, the restricted problem becomes undecidable in simple variants of the basic model, as the following section shows.

\section{Variants of Basic Hexagonal Model}
\label{section:modelVariants}
Our basic hexagonal model can be modified in many ways to model different circumstances, environments or other known problems.

Certainly, the basic model can also be defined for other types of lattices, like the rectilinear square lattice.
Thus, the model covers grid versions of firefighting problems, \eg \cite{article:develin2007fire}, as a special case:
Set $\ignitionCounter[c, 0] = 1$ and $\burningTime[c,0] = \infty$ for each cell $c$ and model the blocking of $c$ at time $t$ via $\ignitionCounter[c, t] = \infty$.
In general, the values for \ignitionCounter and \burningTime could be replaced by positive (not necessarily monotone) functions of the simulation time.

Another natural generalization is to stack several layers of cells on top of each other.
For every cell, the \ignitionCounter- and \burningTime-values could be defined for each layer individually.
These extensions allow to model fire expansion in different heights, such as crown or ground fires.

Environmental factors can also be modelled by slightly adjusting the transition rules.
For example, wind can be modelled by letting a burning cell decrease its neighbours' \ignitionCounter-values by different amounts per round, depending on the direction in which the neighbour lies and in which direction the wind blows.
Cooling down or regrowth of greenery can be modelled by having cells regain their \ignitionCounter- or \burningTime-values if no neighbouring cells are burning.
However, even given regrowth, seemingly similar models, like Conway's Game of Life, remain distinct.

Still, these variants can lead to surprisingly complex problems:
With only three layers of cells in an infinite lattice with a constant description complexity, 
the question if putting fire to cell $c$ will eventually ignite cell $c'$ becomes undecidable.
The problem remains undecidable, even with a single layer, if we allow cells to recover their initial \ignitionCounter-values over time; see Appendix \ref{appendix:mayorsVilla} for details.

\begin{theorem}
\label{theorem:undecidability}
In the version of our firefighting model, where regeneration or at least three layers of cells are allowed, there is no algorithm that can decide every instance of the following problem:

Given a lattice with a finite description, a set of cells $\mathcal{F}$ and a single cell $v$.
Each cell of $\mathcal{F}$ is set on fire at $t=0$.
Will cell $v$ eventually catch fire?
\end{theorem}
\todo{Restrict $\mathcal{F}$ to a single cell in this theorem}

\section{Conclusion}
\label{section:conclusion}
In this paper, we present a new model for firefighting problems together with some solutions and hardness results.
The basic hexagonal model is simple to understand and generalizes a discrete model that has been introduced before.
It allows to incorporate additional parameters to model weather conditions or crown and ground fires.
These extensions could be applied to single cells or the entire lattice.

Obvious questions are how to improve on the results and to widen their scopes.
We did not address any \emph{dynamic} aspects of firefighting, yet.
How does the fire's propagation change, when single cells are fortified?
Moreover, fortifying a path of cells takes time to refill an aircraft's water tanks and fly back and forth.
Can this task be accomplished before the path is reached by the fire?
Research on seemingly simple dynamic geometric problems \cite{article:kim2019geometric,article:klein2019firefighter} seem to indicate that one should not hope for provably optimal results in the basic hexagonal model, but strive for good approximations.

\paragraph{Acknowledgements.}
We thank all anonymous reviewers for their helpful comments and suggestions. 


{
\bibliographystyle{splncs04}
\tiny
\bibliography{HexGrid}

\begin{thebibliography}{10}
\providecommand{\url}[1]{\texttt{#1}}
\providecommand{\urlprefix}{URL }
\providecommand{\doi}[1]{https://doi.org/#1}

\bibitem{book:BoergerEtAl}
B{\"{o}}rger, E., Gr{\"{a}}del, E., Gurevich, Y.: The Classical Decision
  Problem. Perspectives in Mathematical Logic, Springer (1997)

\bibitem{article:bressan2007}
Bressan, A.: Differential inclusions and the control of forest fires. Journal
  of Differential Equations  \textbf{243}(2),  179--207 (2007)

\bibitem{article:develin2007fire}
Develin, M., Hartke, S.G.: Fire containment in grids of dimension three and
  higher. Discrete Applied Mathematics  \textbf{155}(17),  2257--2268 (2007)

\bibitem{book:eiglsperger1999orthogonal}
Eiglsperger, M., Fekete, S.P., Klau, G.W.: Orthogonal graph drawing. In:
  Drawing Graphs, Methods and Models. pp. 121--171 (1999)

\bibitem{survey:finbow2009}
Finbow, S., MacGillivray, G.: The firefighter problem: a survey of results,
  directions and questions. Australasian J. Combinatorics  \textbf{43},  57--78
  (2009)

\bibitem{article:fomin2016firefighter}
Fomin, F.V., Heggernes, P., van Leeuwen, E.J.: The firefighter problem on graph
  classes. Theoretical Computer Science  \textbf{613},  38--50 (2016)

\bibitem{article:gardner1970mathematical}
Gardner, M.: Mathematical games: The fantastic combinations of john conway's
  new solitaire game ``life''. Scientific American  \textbf{223},  120--123
  (1970)

\bibitem{article:garey1977rectilinear}
Garey, M.R., Johnson, D.S.: The rectilinear steiner tree problem in {NP}
  complete. {SIAM} Journal of Applied Mathematics  \textbf{32},  826--834
  (1977)

\bibitem{article:kim2019geometric}
Kim, S., Klein, R., K{\"{u}}bel, D., Langetepe, E., Schwarzwald, B.: Geometric
  firefighting in the half-plane. In: Algorithms and Data Structures - 16th
  International Symposium, {WADS} 2019, Edmonton, AB, Canada, August 5-7, 2019,
  Proceedings, LNCS 11646. pp. 481--494 (2019)

\bibitem{article:klein2019firefighter}
Klein, R., Langetepe, E., Schwarzwald, B., Levcopoulos, C., Lingas, A.: On a
  fire fighter's problem. Int. J. Found. Comput. Sci.  \textbf{30}(2),
  231--246 (2019)

\bibitem{article:minsky1961recursive}
Minsky, M.L.: Recursive unsolvability of post's problem of "tag" and other
  topics in theory of turing machines. Annals of Mathematics  \textbf{74}(3),
  437--455 (1961)

\bibitem{survey:pastor2003mathematical}
Pastor, E., Z{\'a}rate, L., Planas, E., Arnaldos, J.: Mathematical models and
  calculation systems for the study of wildland fire behaviour. Progress in
  Energy and Combustion Science  \textbf{29}(2),  139--153 (2003)

\bibitem{book:toffoli1987cellular}
Toffoli, T., Margolus, N.: Cellular automata machines: a new environment for
  modeling. MIT press (1987)

\bibitem{article:wolfram1983statistical}
Wolfram, S.: Statistical mechanics of cellular automata. Reviews of modern
  physics  \textbf{55}(3), ~601 (1983)

\end{thebibliography}
}

\newpage

\appendix
\section{Appendix}

\subsection{An algorithm for fast fire propagation}
\label{appendix:fastSimulationAlgorithm}
\begin{theorem}
\label{theorem:fastSimulationAlgorithm}
Given a bounded domain $\rectangle$ of $n$ hexagonal cells and a subset of cells $\mathcal{F} \subseteq \rectangle$, which are set on fire at time $t=0$. Let $F(t)$ denote the set of alive cells adjacent to at least one burning cell at time $t$.
The result of the fire propagation in the grid can be computed in \bigO{D_{end} ~\cdot~ \log F_{max}} = \bigO{n \cdot \log n} time, where $F_{max} := \max_{t\geq 0} \abs{F(t)}$ and $D_{end}$ denotes the number of cells that are dead when no burning cell remains.
\end{theorem}

Let $\mathcal{N}_c$ be the set of neighbours of a cell $c$.
Let \emph{ignition time} \ignitionTime[c] denote the moment in time when $c$ ignites, \ie the minimum $t$ for which $\ignitionCounter[c,t]= 0$ holds, or $\infty$ if no such $t$ exists.
More precisely, if such a value for $t$ exists, then
\[
\ignitionTime[c] =
\min \left\lbrace t\in \mathbb{N} ~\left|~ \ignitionCounter[c,0] \leq \sum\limits_{n \in \mathcal{N}_c} \max (0, \min(t-\ignitionTime[n], \burningTime[n,0]))\right.\right\rbrace.
\]
Having computing \ignitionTime[c] for all $c$ of the grid, it is possible to return the result of the fire propagation at any time $t$.
All cells $c$ with $\ignitionTime[c]=\infty$ never ignite.

\begin{lemma}
Given $\ignitionTime[n]$ for all $n \in \mathcal{N}_c$, $\ignitionTime[c]$ can be calculated in constant time.
\end{lemma}

\begin{proof}
The time during which a neighbour $n$ burns corresponds to an interval $[\ignitionTime[n], \ignitionTime[n] + \burningTime[n,0]]$, as illustrated in \autoref{figure:predictIgnition}.
As the number of neighbours is constant (at most $6$), we can sweep over the constant number of start and endpoints of these intervals in ascending order and calculate $\ignitionCounter[c,t]$ at each such point, until it reaches $0$ (at some time $t^+$).
The last interval start or endpoint before that is $t^-$.
Between $t^+$ and $t^-$, $\ignitionCounter[c]$ decreases by a constant amount each round, so we can directly calculate $\ignitionTime[c]$.
If we reach the last interval endpoint without finding $t^+$, then $\ignitionTime[c] = \infty$.
\end{proof}

Now, let $N \subseteq \mathcal{N}$ be a subset of the neighbours of $c$ for which \ignitionTime is known.
We define the \emph{predicted ignition time} \predictedIgnitionTime[N,c] to be the time at which $c$ would ignite if the cells in $N$ were the only neighbours of $c$.
That is to say, we can define and compute $\predictedIgnitionTime[N,c]$ in the same way as $\ignitionTime[c]$, just with a subset of $\mathcal{N}_c$.
As we are working with a subset of the actual neighbours of $c$, it is clear, that $\predictedIgnitionTime[N,c] \geq \ignitionTime[c]$ for any $N \subseteq \mathcal{N}$.
Our algorithm is mostly based on the following intuitive observation:
A neighbour $n$ of $c$ that ignites after $c$ does not actually affect \ignitionTime[c].

\begin{lemma}
	\label{lemma:predictedIgnition}
 	Let $c$ be a cell and $\mathcal{N}_c$ its neighbours.
 	If $N = \{n \in \mathcal{N}_c \mid \ignitionTime[n] < \ignitionTime[c]\}$, then $\predictedIgnitionTime[N,c] = \ignitionTime[c]$ holds.
\end{lemma}
\begin{proof}
	Obviously, if $N = \mathcal{N}_c$, then $\predictedIgnitionTime[N,c] = \ignitionTime[c]$.
	Consider a cell $n_i \in \mathcal{N}_c$ that is not in $N$;
	hence $\ignitionTime[n_i] \geq \ignitionTime[c] \Leftrightarrow \ignitionTime[c]-\ignitionTime[n_i] \leq 0$
	holds and consequently $\max \left(0, \min \left(\ignitionTime[c]-\ignitionTime[n], \burningTime[n,0]\right)\right) = 0$.
	This implies that $n_i$ does not affect $\ignitionTime[c]$:
	\begin{align*}
	\ignitionTime[c] &= \min \left\lbrace t\in\mathbb{N} ~\left|~ \ignitionCounter[c,0] \leq \sum\limits_{n \in \mathcal{N}_c} \max \left(0, \min \left(t-\ignitionTime[n], \burningTime[n,0]\right)\right)\right.\right\rbrace\\
	&= \min \left\lbrace t\in\mathbb{N} ~\left|~ \ignitionCounter[c,0] \leq \sum\limits_{n \in \mathcal{N}_c\setminus\{n_i\}} \max \left(0, \min \left(t-\ignitionTime[n], \burningTime[n,0]\right)\right) \right.\right\rbrace\\
	&= \predictedIgnitionTime[c, \mathcal{N}_c\setminus \{n_i\}].
	\end{align*}
\end{proof}

\begin{figure}[bt]
\centering
\begin{tikzpicture}
\draw[thick, |->] (0,0) node[left] (start) {0} -- (11,0) node[below] (end) {$t$};
\draw[|-|] (1, 1.5) node[left] (I1start) {\ignitionTime[n_1]} to node[above, midway] {\burningTime[n_1,0]} (4, 1.5) node[right] (I1end) {};
\draw[|-|] (2, 1  ) node[left] (I2start) {\ignitionTime[n_2]} to node[above, midway] {} (8, 1.0) node[right] (I2end) {};
\draw[|-|] (3, 0.5) node[left] (I3start) {\ignitionTime[n_3]} to node[above, midway] {} (5, 0.5) node[right] (I3end) {};
\draw[dashed, |-] (8.5, 1.5) node[left] (I4start) {} to node[midway, above] {\ignitionTime[n_4]} (end |- {{(0,1.5)}}) node[right] (I4end) {};
\draw[dashed, |-] (9.5, 1.0) node[left] (I5start) {} -- (end |- {{(0,1.0)}}) node[right] (I5end) {};
\draw[dashed, |-] ( 10, 0.5) node[left] (I4start) {} -- (end |- {{(0,0.5)}}) node[right] (I4end) {};

\draw[ultra thick] (7, 2) node[left] {$\ignitionTime[c]$} -- (7, -.5) node[right] {$\sum = 10 $};
\draw[dotted] (I1start.east) -- (I1start.east |- {{(0, -.25}});
\draw[dotted] (I2start.east) -- (I2start.east |- {{(0, -.25}});
\draw[dotted] (I3start.east) -- (I3start.east |- {{(0, -.25}});
\draw[dotted] (I1end.west) -- (I1end.west |- {{(0, -.25}});
\draw[dotted] (I2end.west |- {{(0, +2)}}) node[left] {$t^+$} -- (I2end.west |- {{(0, -.25}});
\draw[dotted] (I3end.west |- {{(0, +2)}}) node[left] {$t^-$} -- (I3end.west |- {{(0, -.25}});
\node[below] at (0.5, 0) {0};
\node[below] at (1.5, 0) {1};
\node[below] at (2.5, 0) {2};
\node[below] at (3.5, 0) {3};
\node[below] at (4.5, 0) {2};
\node[below] at (6.5, 0) {1};
\end{tikzpicture}
\caption{
A sweep to compute $\ignitionTime[c]$ in constant time.
Solid intervals indicate neighbours that affect $\predictedIgnitionTime[N,c]$, while dashed intervals indicate neighbouring cells that do not.
If $\ignitionCounter[c] = 10$, via sweep from left to right we can find $\ignitionTime[c] = 1 \cdot 1 + 1 \cdot 2 + 1 \cdot 3 + 1 \cdot 2 + 2 \cdot 1 = 10$ in constant time.
}
\label{figure:predictIgnition}
\end{figure}
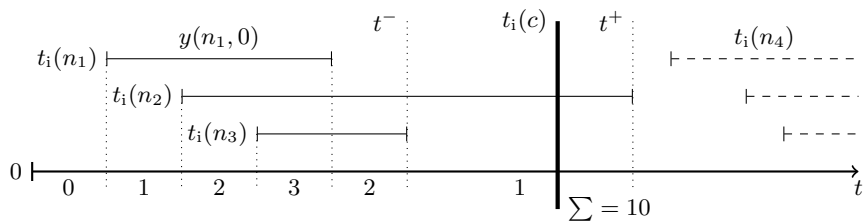

The idea of \autoref{algorithm:fastSimulation} is to compute \ignitionTime of all cells in ascending order;
a min-ordered priority-queue \queue is used to store, update tuples $(\predictedIgnitionTime, c)$ and to extract them when $\predictedIgnitionTime = \ignitionTime[c]$.
The priority-queue provides the following two operations:
\Update{\predictedIgnitionTime', c}, which checks whether \queue contains a tuple $(\predictedIgnitionTime, c)$ and inserts $(\predictedIgnitionTime', c)$ if this is not the case;
however, if $(\predictedIgnitionTime, c)\in \queue$, the operation updates the tuple to $(\min \lbrace \predictedIgnitionTime, \predictedIgnitionTime' \rbrace, c)$.
\ExtractMin{} removes a tuple $(\predictedIgnitionTime, c)\in \queue$ from \queue where \predictedIgnitionTime is minimal over all tuples of \queue and returns it.

\begin{procedure}[bt]
	\caption{FastFirePropagation($\mathcal{F}$)}
	\label{algorithm:fastSimulation}
	\Input{Domain \rectangle and set of initially burning cells $\mathcal{F} \subset \rectangle$.}
    \Output{$\ignitionTime[c]$ is computed for all $c\in\rectangle$ with $\ignitionTime[c]<\infty$.}
	\BlankLine
	$C = \emptyset$\\
	\ForEach{cell $c \in \mathcal{F}$}{
		\Update{0, c}\\
	}
	\While{\queue has node with finite value for \predictedIgnitionTime}{
		$(\predictedIgnitionTime, c) \gets \ExtractMin{}$\\
		$\ignitionTime[c] \gets \predictedIgnitionTime$\\
		$C \gets C \cup \lbrace c\rbrace$\\
		\ForEach{neighbour $n \in \mathcal{N}_c$ with $n \notin C$}{
			$\predictedIgnitionTime' \gets \predictedIgnitionTime[\mathcal{N}_n \cap C, n]$\\
			\If{$\predictedIgnitionTime' < \infty$}{
				\Update{\predictedIgnitionTime', n}\\
			}
		}
	}
\end{procedure}

We prove the correctness of this algorithm via the following loop invariant.
As the algorithm terminates when $Q$ is empty, Invariant \ref{li1} and Invariant \ref{li3} together guarantee, that the algorithm correctly calculates $\ignitionTime[c]$ for all $c$.

\begin{lemma}[loop invariant]
\label{lemma:loopInvariant}
After the $k$-th iteration of the while-loop, the following conditions hold:
\begin{enumerate}
	\item\label{li1} For all $c \in C$, \ignitionTime[c] has been computed correctly;
	\item\label{li2} for any cell $c \in C$ and any $c' \notin C$, it is $\ignitionTime[c] \leq \ignitionTime[c']$;
	\item\label{li3} for all cells $c \notin C$, if $(\predictedIgnitionTime,c) \in \queue$, it is $\predictedIgnitionTime =\predictedIgnitionTime[\mathcal{N}_c \cap C, c]$; otherwise $\predictedIgnitionTime[\mathcal{N}_c \cap C, c] = \infty$;
	\item\label{li4} for any tuple $(\predictedIgnitionTime,c) \in \queue$ where \predictedIgnitionTime is minimal among all tuples in \queue, it is $\predictedIgnitionTime = \ignitionTime[c]$ and $\ignitionTime[c]$ is minimal among all cells $\notin C$.
\end{enumerate}
\end{lemma}
\begin{proof}[Proof (by induction).]
	After the initialization step, $C$ is still empty; hence, part \ref{li1} and \ref{li2} of the invariant trivially hold.
	In addition, \queue contains a tuples $(0,c)$ for each $c \in \mathcal{F}$, which obviously satisfies part \ref{li4};
	moreover, for cells $c \not\in \mathcal{F}$ it is $\predictedIgnitionTime[\mathcal{N}_c \cap C, c] = \infty$, since $C = \emptyset$, which together satisfies part \ref{li3}.
	
	Assume that the invariant holds for $(k-1)$ and consider the $k$-th iteration:
	
	A tuple $(\predictedIgnitionTime,c) \in \queue$ with a minimal \predictedIgnitionTime among all tuples in \queue is extracted, $c$ is added to $C$ and $\ignitionTime[c] = \predictedIgnitionTime$ is set.
	With part \ref{li4} of the induction hypothesis, we obtain that part \ref{li1} and \ref{li2} still hold after the $k$-th iteration.

	Let $c'$ be any cell $\notin C$.
	If $c'$ is not a neighbour of $c$, part \ref{li3} still holds for $c'$ by induction hypothesis since $\mathcal{N}_{c'} \cap C = \mathcal{N}_{c'} \cap (C \cup \{c\})$;
	otherwise part \ref{li3} holds for $c'$ as its tuple gets updated during the while-loop.
	
	It remains to show that part \ref{li4} still holds.
	Given $(\predictedIgnitionTime[], c^*) \in \queue$ where \predictedIgnitionTime[] is minimal in \queue after the $k$-th iteration. 
	If $\predictedIgnitionTime[]=0$, the invariant holds; hence, we may assume $\predictedIgnitionTime[]>0$.
    Further assume there is a cell $c \notin C$ such that $\ignitionTime[c] < \predictedIgnitionTime[]$ and $c$ has a minimal \ignitionTime[c] among all cells not in $C$.
    Due to the fact that part~\ref{li3} holds, we may conclude that $(\predictedIgnitionTime[\mathcal{N}_c \cap C], c)\in\queue$.%
    \footnote{Otherwise $\predictedIgnitionTime[\mathcal{N}_c \cap C, c] = \infty$, which would imply that there is a cell $c' \in\mathcal{N}_c$ such that $\ignitionTime[c'] < \ignitionTime[c]$ which contradicts the choice of $c$.}
    By \autoref{lemma:predictedIgnition} and minimality of $(\predictedIgnitionTime[], c^*)$ in \queue, we obtain $\predictedIgnitionTime[\lbrace n \in \mathcal{N}_c ~|~ \ignitionTime[n] < \ignitionTime[c]\rbrace, c] = \ignitionTime[c] < \ignitionTime[c^*] \leq \predictedIgnitionTime[] \leq \predictedIgnitionTime[\mathcal{N}_c \cap C, c]$,
    which gives a contradiction.
    Hence, part~\ref{li4} holds.
%
\end{proof}

The number of iterations is bounded by the number of cells that enter \queue, since one cell is extracted in each iteration;
only cells with $\ignitionTime[c] <\infty$ ever enter \queue, which bounds it by $D_{end}$, which is the number of dead cells at the end of the simulation.
In each iteration, $C$ contains contains cells with $\ignitionTime[c] < t_j$ for some $t_j$ and \queue contains cells that neighbour cells in $C$;
therefore the size of $\queue$ after the iteration is bounded by $F(t_j)$ and by $F_{max} := \max_{t\geq 0} \abs{F(t)}$ overall.
Hence the \autoref{algorithm:fastSimulation} runs in \bigO{D_{end} \cdot \log F_{max}}= \bigO{n \log n}, which completes the proof of \autoref{theorem:fastSimulationAlgorithm}.


\FloatBarrier
\newcommand{\HFire}{\ensuremath{\mathcal{R}_{\textit{Fire}}}\xspace}
\newcommand{\HTwoRM}{\ensuremath{\mathcal{R}_{2RM}}\xspace}
\subsection{Proof of \autoref{theorem:undecidability}}
\label{appendix:mayorsVilla}
To prove the theorem, we show how to simulate computations of so-called \emph{two-register machines} \cite{book:BoergerEtAl,article:minsky1961recursive} with our fire expansion model.
A two-register machine $M$ has exactly two registers $r_1, r_2$ that may contain an arbitrarily large natural number each.
To manipulate the content of the registers, there are two types of instructions.

An \emph{addition instruction} is a triple $(i, r, j)$ where $i \neq j$ is the number of the instruction, $r \in \lbrace r_1, r_2 \rbrace$ is a register and $j$ is the number of the following instruction.
When $M$ processes the instruction, the value of register $r$ is incremented by one.
Afterwards, $M$ continues with instruction number $j$.

A \emph{subtraction instruction} is a quadruple $(i, r, j, k)$ where $j \neq i \neq k$ are numbers of instructions and $r \in \lbrace r_1, r_2 \rbrace$ is a register.
When processing the instruction, $M$ checks whether the content of register $r$ is zero:
If this is the case, then $M$ proceeds with instruction $I_j$;
otherwise, $M$ decrements $r$ by one and proceeds with instruction $I_k$.
Note that using subtraction instructions, it is possible to branch the control flow and simulate conditional statements.

In the following, with $I_i$ we refer to the instruction with number $i$ regardless of its type.
A \emph{program} for two-register machine is a finite sequence of instructions $(I_1, \ldots, I_n)$; we assume $I_n$ to be a special instruction that causes the machine to halt.
With $M$ we refer to the program, as well as the two-register machine.

A \emph{configuration} $C$ of a two-register machine $M$ is a triple $(d_1, d_2, I_i)$ where $d_1, d_2 \in \mathbb{N}$ are the contents of the registers $r_1, r_2$ and $I_i$ is the instruction that $M$ is going to process next.

\newcommand{\geomCell}{box\xspace}
\newcommand{\geomCells}{boxes\xspace}
\paragraph{Geometric interpretation.}
For a two-register machine $M$ and its program, there is a nice geometric interpretation.
Imagine the first quadrant of the plane is divided into square \geomCells\footnote{We use the term \emph{\geomCell} to underline the difference to \emph{cells} which contain the \ignitionCounter and \burningTime-values in our model.} by the $\mathbb{N}\times\mathbb{N}$ grid; see \autoref{figure:geometricInterpretation}.
Each of these boxes can be uniquely described by a tuple of integers coordinates $(r_1,r_2)$, where $(0,0)$ corresponds to the bottom left most \geomCell.
Respectively, the coordinates of each \geomCell can be interpreted as the content of the registers of $M$;
\ie, the \geomCell with coordinates $(3, 2)$ corresponds to the register contents $3$ and $2$.

Into the quadrant of the plane, we embed a directed graph $G_M = (V, E)$ as follows:
In each \geomCell, a vertex $v_i$ is created for each instruction $I_i$ and placed along the diagonal from the upper-left to the bottom-right corner.
To distinguish the vertices of a \geomCell $c$ from those of a neighbouring \geomCell $n$, we use a superscript and write $v_i^c$ and $v_i^n$.
The directed edges of the graph run between vertices of the same or a neighbouring \geomCell and are introduced for each \geomCell $c$ according to the following rules:
\todo{to adress a vertex in a box with coordinates $(c, d)$, we use superscript ...}
\DK{use the introduced notation for the definition of the edges.}
\begin{itemize}
\item For each addition instruction $I_i = (i, r, j)$ of $M$, insert a directed edge $(v_i^c, v_j^n)$ where $n$ is the right (upper) neighbour of $c$ if $r=r_1$ $(r=r_2)$.
\item For each subtraction instruction $I_i = (i, r, j, k)$ of $M$, insert a directed edge between nodes of neighbouring \geomCells considering the following cases:
\begin{itemize}
	\item Let $r=r_1$. If $c$ has a left neighbour $n$, then add an edge $(v_i^c, v_j^n)$; otherwise, add the edge $(v_i^c, v_k^u)$ where $u$ denotes the \geomCell above $c$.
	\item Let $r=r_2$. If $c$ has a lower neighbour $n$, then add an edge $(v_i^c, v_j^n)$; otherwise, add the edge $(v_i^c, v_k^r)$ where $r$ denotes the \geomCell right of $c$.
\end{itemize} 
\end{itemize}
\begin{figure}[tbh]
\centering
\begin{tikzpicture}
\colorlet{grey}{black!30}
\pgfmathsetmacro{\Gridsize}{2.0}
\draw[step=\Gridsize] (0,0) grid (5*\Gridsize,3*\Gridsize);
\draw[->] (0,0) -- (5*\Gridsize + .5,0) node[below] {$r_1$};
\draw[->] (0,0) -- (0,3*\Gridsize + .5) node[left] {$r_2$};
\foreach \x in {0,...,4}
	\node[below] at (+\Gridsize/2+\Gridsize*\x,0) {$\x$};
\foreach \y in {0,...,2}
	\node[left] at (0,+\Gridsize/2+\Gridsize*\y) {$\y$};

\foreach \row in {1,...,4} {
    \foreach \column in {1,...,5} {
		\coordinate (ir\row;ic\column) at (-1.5+\Gridsize*\column,-0.5+\Gridsize*\row);
		\coordinate (jr\row;jc\column) at (-1.0+\Gridsize*\column,-1.0+\Gridsize*\row);
		\coordinate (kr\row;kc\column) at (-0.5+\Gridsize*\column,-1.5+\Gridsize*\row);
	}
}
\draw[grey, fill=grey] (ir2;ic5) circle (2pt) node[above right] {$v_i$};
\draw[grey, fill=grey] (jr2;jc5) circle (2pt) node[above right] {$v_j$};
\draw[grey, fill=grey] (kr2;kc5) circle (2pt) node[above right] {$v_k$};

\draw[black, fill=black] (ir2;ic4) circle (2pt) node[below right] {$v_i$};
\draw[black, fill=black] (jr2;jc4) circle (2pt) node[below right] {$v_j$};
\draw[grey, fill=grey] (kr2;kc4) circle (2pt) node[below right] {$v_k$};

\draw[grey, fill=grey] (ir2;ic3) circle (2pt) node[below left] {$v_i$};
\draw[grey, fill=grey] (jr2;jc3) circle (2pt) node[below left] {$v_j$};
\draw[black, fill=black] (kr2;kc3) circle (2pt) node[below left] {$v_k$};

\draw[grey, fill=grey] (ir1;ic4) circle (2pt) node[below left] {$v_i$};
\draw[grey, fill=grey] (jr1;jc4) circle (2pt) node[below left] {$v_j$};
\draw[grey, fill=grey] (kr1;kc4) circle (2pt) node[below left] {$v_k$};

\draw[grey, fill=grey] (ir3;ic4) circle (2pt) node[below right] {$v_i$};
\draw[black, fill=black] (jr3;jc4) circle (2pt) node[below right] {$v_j$};
\draw[grey, fill=grey] (kr3;kc4) circle (2pt) node[below right] {$v_k$};

\draw[bend right, thick, -latex] (ir2;ic4) to [bend left] (jr3;jc4);
\draw[bend right, thick, -latex] (jr2;jc4) to [bend left] (kr2;kc3);

\draw[grey, fill=grey] 	(ir3;ic1) circle (2pt) node[below right] {$v_i$};
\draw[black, fill=black] 	(jr3;jc1) circle (2pt) node[below right] {$v_j$};
\draw[grey, fill=grey] 	(kr3;kc1) circle (2pt) node[below right] {$v_k$};

\draw[black, fill=black] 	(ir2;ic1) circle (2pt) node[below left] {$v_i$};
\draw[black, fill=black] 	(jr2;jc1) circle (2pt) node[below left] {$v_j$};
\draw[grey, fill=grey] 	(kr2;kc1) circle (2pt) node[below left] {$v_k$};

\draw[grey, fill=grey] (ir2;ic2) circle (2pt) node[above right] {$v_i$};
\draw[grey, fill=grey] (jr2;jc2) circle (2pt) node[above right] {$v_j$};
\draw[grey, fill=grey] (kr2;kc2) circle (2pt) node[above right] {$v_k$};

\draw[grey, fill=grey] (ir1;ic1) circle (2pt) node[below left] {$v_i$};
\draw[grey, fill=grey] (jr1;jc1) circle (2pt) node[below left] {$v_j$};
\draw[grey, fill=grey] (kr1;kc1) circle (2pt) node[below left] {$v_k$};

\draw[bend right, thick, -latex] (ir2;ic1) to [bend left] (jr3;jc1);
\draw[bend right, thick, -latex] (jr2;jc1) to [bend right] (ir2;ic1);
\end{tikzpicture}
\caption{
In the geometric interpretation of a two-register machine, a vertex is created for each instruction and in every \geomCell;
for simplicity, this has only been done for \geomCells with coordinates $(4,2)$, $(1,2)$ and their neighbouring \geomCells in this sketch.
For the \geomCells with coordinates $(4,2)$ and $(1,2)$, the directed edges are given that result from the addition instruction $I_i = (i, r_2, j)$ and the subtraction instruction $I_j = (j, r_1, i, k)$.
}
\label{figure:geometricInterpretation}
\end{figure}
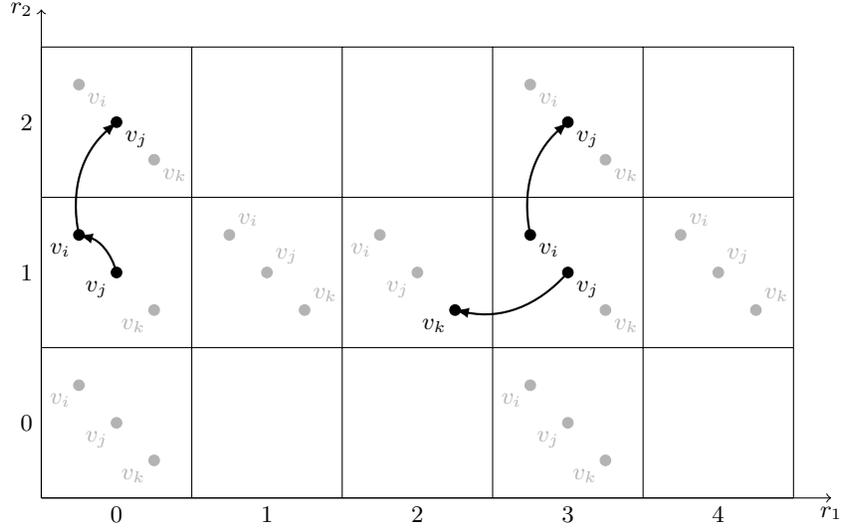

Note that there are (countably) infinite many \geomCells, vertices and edges embedded in the plane, but the description itself is finite.
Clearly, there is a one-to-one correspondence between vertices of $G_M$ and configurations of $M$:
A vertex $v_i$ in a \geomCell with coordinates $(d_1, d_2)$ corresponds to the configuration $(d_1, d_2, I_i)$ of the two-register machine $M$ and vice versa.
\todo{use superscripts}
Similarly, a sequence of configurations of $M$ corresponds to a unique directed path in $G_M$ and vice versa.

\paragraph{Simulation of two-register machines.}
Computations of a two-register machine can be simulated by expanding fire in an orthogonal square lattice.
A very basic observation is that, with an appropriate choice of weights, we can force a fire to expand along a sequence of cells into a certain direction; see \autoref{figure:wire}.
In the following, we call such a sequence of cells a \emph{wire}.
\begin{figure}[tbh]
\centering
	\begin{subfigure}[b]{0.3\textwidth}
	    \centering
    	\begin{tikzpicture}
    	\grid{3}{7}{}{burntSquare};
        \node[aliveSquare] at (r1;c2) {1;1};
        \node[aliveSquare] at (r2;c2) {2;1};
        \node[aliveSquare] at (r2;c3) {1;1};
        \node[burningSquare] at (r2;c4) {0;1};
        \node[aliveSquare] at (r2;c5) {1;2};
        \node[aliveSquare] at (r2;c6) {2;1};
        \node[aliveSquare] at (r2;c7) {1;1};
		\end{tikzpicture}
		\caption{Initial fire at time $t$.}
	\end{subfigure}
	~
	\begin{subfigure}[b]{0.3\textwidth}
	    \centering
    	\begin{tikzpicture}
    	\grid{3}{7}{}{burntSquare};
        \node[aliveSquare] at (r1;c2) {1;1};
        \node[aliveSquare] at (r2;c2) {2;1};
        \node[burningSquare] at (r2;c3) {0;1};
        \node[burntSquare] at (r2;c4) {0;0};
        \node[burningSquare] at (r2;c5) {0;2};
        \node[aliveSquare] at (r2;c6) {2;1};
        \node[aliveSquare] at (r2;c7) {1;1};
		\end{tikzpicture}
		\caption{At time $t+1$.}
	\end{subfigure}
    ~
	\begin{subfigure}[b]{0.3\textwidth}
	    \centering
    	\begin{tikzpicture}
    	\grid{3}{7}{}{burntSquare};
        \node[aliveSquare] at (r1;c2) {1;1};
        \node[aliveSquare] at (r2;c2) {1;1};
        \node[burntSquare] at (r2;c3) {0;0};
        \node[burntSquare] at (r2;c4) {0;0};
        \node[burntSquare] at (r2;c5) {0;0};
        \node[burntSquare] at (r2;c6) {0;0};
        \node[burningSquare] at (r2;c7) {0;1};
		\end{tikzpicture}
		\caption{At time $t+4$.}
	\end{subfigure}
\caption{Fire expanding along a wire. For all grey cells, we assume that $\ignitionCounter = \burningTime = 0$. Note that due to choices for \ignitionCounter and \burningTime, after a while the fire can only continue to expand to the right.}
\label{figure:wire}
\end{figure}
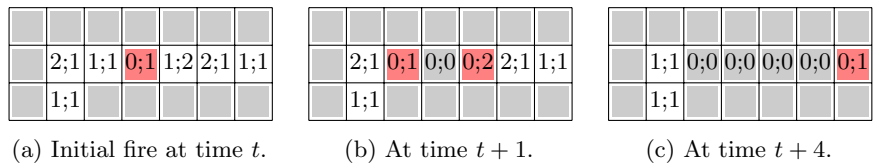
We now describe the layout and the cells of the orthogonal square lattice.
We start with a subdivision of the first quadrant of the plane into \geomCells as described above and shown in \autoref{figure:geometricInterpretation}.
To obtain the cells of the firefighting model, we refine the \geomCells using the following gadgets:
\begin{itemize}
\item The \textit{crossing gadget} allows a fire to burn along one wire and cross another wire without causing it to burn;
\item the \textit{instruction gadget} implements a vertex $v_i$ that corresponds to an instruction $I_i$ in a \geomCell;
\item the \textit{\geomCell gadget} defines the basic layout of every \geomCell with the use of crossing and instruction gadgets;
\item the \textit{connector gadget} is used to link neighbouring \geomCell gadgets.
\end{itemize}
In the standard version of our model, it seems impossible to realize a \emph{crossing gadget} since a fire burning along one wire tends to expand also along the other one.
However, if we allow (i) either for a constant number of additional layers or (ii) regeneration of the \ignitionCounter-values of cells, crossings can in fact be realized:
(i) With three layers of cells, it is not difficult to realize a crossing within a $5\times 5 \times 3$ cuboid as shown in \autoref{subfigure:crossingMultiLayer}.
While one of the wires runs straight through the bottom layer, the second one bypasses it using the topmost layer such that both wires are separated by one \dead cell.
(ii) When cells are allowed to restore their original resistance values \ignitionCounter over time, a crossing can be realized even in a two-dimensional lattice using $11 \times 11$ cells.
The gadget shown in \autoref{subfigure:crossingGadgetRegeneration} is based on the assumption that the following rule is added to the definition of the basic hexagonal model in \autoref{section:introduction}:
\textit{If $c$ is not \burning at time $t$ and $t+1$ and has no direct burning neighbour at time $t$, then $\ignitionCounter[c, t+1]:= \min\lbrace \ignitionCounter[c, 0], \ignitionCounter[c,t]+1 \rbrace$ and $\burningTime[c, t+1]:= \min\lbrace \burningTime[c, 0], \burningTime[c,t]+1 \rbrace$.}

\begin{figure}[tbh]
    \centering
    \begin{subfigure}[b]{0.45 \textwidth}
        \centering
        \begin{tikzpicture}[scale=0.7]
            \clip (-3,-3) rectangle (3,3);
\coordinate (tf) at (0,0);
\coordinate (bf) at (0,-3);
\coordinate (tr) at (15:3.0cm);
\coordinate (tl) at (165:3.0cm);

   \coordinate (fr) at ($ (tf)!5!(tr) $);
   \coordinate (fl) at ($ (tf)!5!(tl) $);
   \coordinate (fb) at ($ (tf)!15!(bf) $);

   \path[name path=brpath] (bf) -- (fr);
   \path[name path=rbpath] (tr) -- (fb);
   \path[name path=blpath] (bf) -- (fl);
   \path[name path=lbpath] (tl) -- (fb);
   \path[name path=trpath] (tl) -- (fr);
   \path[name path=tlpath] (tr) -- (fl);

   \draw[name intersections={of=brpath and rbpath}] (intersection-1)coordinate (br){}; 
   \draw[name intersections={of=blpath and lbpath}] (intersection-1)coordinate (bl){}; 
   \draw[name intersections={of=trpath and tlpath}] (intersection-1)coordinate (tb){}; 

\draw[opacity=0.5,fill=gray] (tf) -- (bf) -- (bl) -- (tl) -- cycle;
\draw[opacity=0.5,fill=gray] (tf) -- (bf) -- (br) -- (tr) -- cycle;
\draw[opacity=0.5,fill=gray] (tf) -- (tr) -- (tb) -- (tl) -- cycle;

\def\tone{.23} 
\def\ttwo{.44}
\def\tthree{.64}
\def\tfour{.82}
\def\fone{.36}
\def\ftwo{.7}
\draw[fill=white] ($ (bf)!\ttwo!(br) $)  -- ($ (tf)!\ttwo!(tr) $) -- ($ (tl)!\ttwo!(tb) $) -- ($ (tl)!\tthree!(tb) $) -- ($ (tf)!\tthree!(tr) $) -- ($ (bf)!\tthree!(br) $) -- cycle;
\draw[fill=white] let \p1=($(tf)!\ftwo!(bf)$), \p2=($(tl)!\ftwo!(bl)$) in ($ (bf)!\ttwo!(bl) $) -- ($ (bf)!\tthree!(bl) $)  -- ($(\p1)!\tthree!(\p2)$) -- ($(\p1)!\ttwo!(\p2)$) -- cycle;

\draw (tf) -- (bf);
\draw (tf) -- (tr);
\draw (tf) -- (tl);
\draw (tr) -- (br);
\draw (bf) -- (br);
\draw (tl) -- (bl);
\draw (bf) -- (bl);
\draw (tb) -- (tr);
\draw (tb) -- (tl);

\draw ($ (bf)!\tone!(br) $)   -- ($ (tf)!\tone!(tr) $)   -- ($ (tl)!\tone!(tb) $);
\draw ($ (bf)!\ttwo!(br) $)   -- ($ (tf)!\ttwo!(tr) $)   -- ($ (tl)!\ttwo!(tb) $);
\draw ($ (bf)!\tthree!(br) $) -- ($ (tf)!\tthree!(tr) $) -- ($ (tl)!\tthree!(tb) $);
\draw ($ (bf)!\tfour!(br) $)  -- ($ (tf)!\tfour!(tr) $)  -- ($ (tl)!\tfour!(tb) $);

\draw ($ (bf)!\tone!(bl) $)   -- ($ (tf)!\tone!(tl) $)   -- ($ (tr)!\tone!(tb) $);
\draw ($ (bf)!\ttwo!(bl) $)   -- ($ (tf)!\ttwo!(tl) $)   -- ($ (tr)!\ttwo!(tb) $);
\draw ($ (bf)!\tthree!(bl) $) -- ($ (tf)!\tthree!(tl) $) -- ($ (tr)!\tthree!(tb) $);
\draw ($ (bf)!\tfour!(bl) $)  -- ($ (tf)!\tfour!(tl) $)  -- ($ (tr)!\tfour!(tb) $);

\draw ($ (tl)!\fone!(bl) $) -- ($ (tf)!\fone!(bf) $) -- ($ (tr)!\fone!(br) $);
\draw ($ (tl)!\ftwo!(bl) $) -- ($ (tf)!\ftwo!(bf) $) -- ($ (tr)!\ftwo!(br) $);
        \end{tikzpicture}
        \caption{A crossing gadget using three layers of cells.}
        \label{subfigure:crossingMultiLayer}
    \end{subfigure}
    ~
    \begin{subfigure}[b]{0.45 \textwidth}
        \centering
        \begin{tikzpicture}[scale=0.95]
            \input{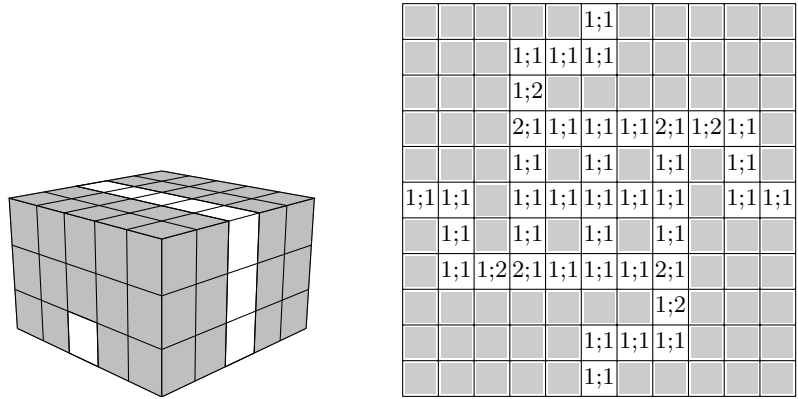}
        \end{tikzpicture}
        \caption{A crossing of wires based on regeneration of \ignitionCounter-values.}
        \label{subfigure:crossingGadgetRegeneration}
    \end{subfigure}
    \caption{Realization of a crossing gadget depending on the variant of the model.}
\end{figure}

The layout of an \emph{instruction gadget} is given in \autoref{figure:instructionGadget}.
Four cells with weights $(1;2)$ along the boundary mark the positions where incoming wires will be connected to;
hence, we call them \emph{input cells}.
Note that, due to the choice of these weights, a fire cannot ignite an input cell (and an attached wire) from inside the gadget.
Moreover, a fire burning along an incoming wire will eventually reach the \emph{central cell} with weights $(1;1)$.
Afterwards, the fire will continue to burn along the wire and reach the cell with weights $(1;1)$ at the boundary which we call the \emph{output cell}.
To obtain the instruction gadget for a specific instruction, we rotate the gadget by a multiple of $90^\circ$ such that the output cell points into the direction of the involved neighbouring \geomCell:
For addition instructions to $r_1$ ($r_2$), the output cell points to the right (upwards);
for a subtraction instruction to $r_1$ ($r_2$), the output cell points to the left (downwards).
\begin{figure}[bt]
    \centering
    \begin{tikzpicture}[scale=0.95]
        \input{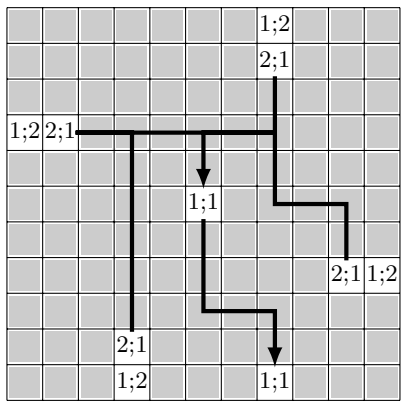}
    \end{tikzpicture}
    \caption{
    The layout of an instruction gadget.
    Black edges indicate cells with weights $\ignitionCounter =\burningTime = 1$; all remaining grey cells have $\ignitionCounter =\burningTime =0$.
    The choice of weights assures two properties:
    a fire reaching an \emph{input cell} at the boundary can expand into the gadget;
    a fire inside the gadget can only reach the \emph{output cell} on the boundary.
    }
    \label{figure:instructionGadget}
\end{figure}

Given the layout of both, instruction and crossing gadgets, we now describe how a \emph{\geomCell gadget} is constructed.
Let $n$ be the number of instructions of $M$.
Start with an orthogonal square lattice of $23n \times 23n$ cells, where all cells get the weights $\ignitionCounter = \burningTime = 0$.
In a first step, we place $n$ instruction gadgets along the diagonal from the top left to the bottom right cell.
Then, we add (in- and output) wires and crossing gadgets as shown in \autoref{figure:boxGadgets}:
In each direction and for each instruction gadget, the \geomCell gadget has an input and output wire, which can be used to connect the instruction gadget to a neighbouring \geomCell into this direction.
\begin{figure}[bt]
    \centering
        \includegraphics[width=.5 \textwidth, page=4]{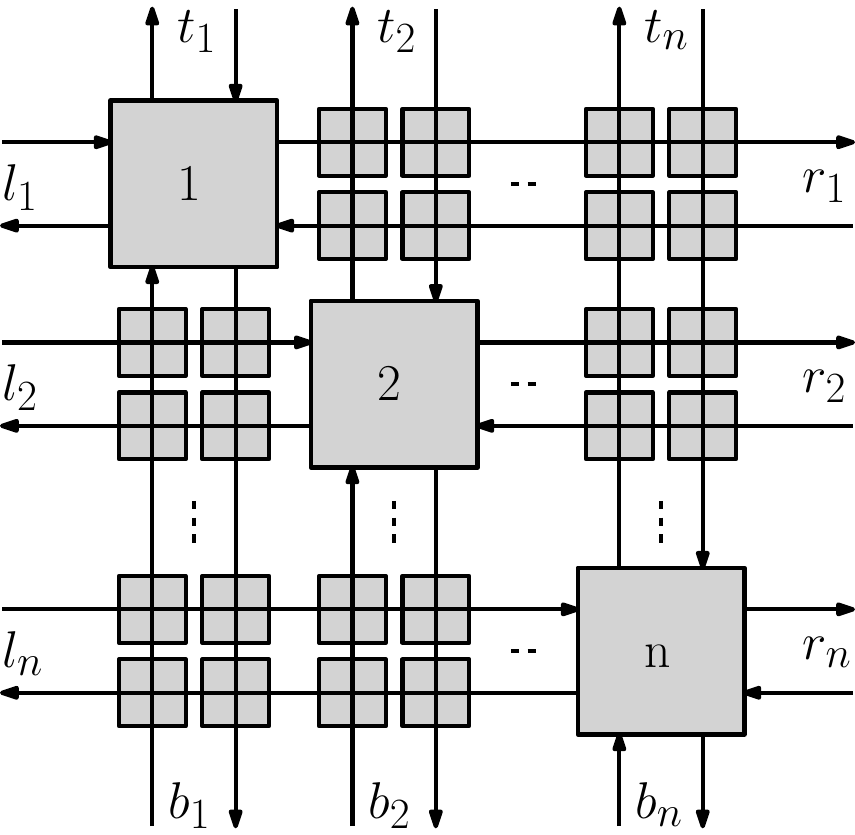}
    \caption{
    Layout of the \geomCell gadget:
    Large rectangles along the diagonal indicate instruction gadgets of size $23 \times 23$;
    small rectangles indicate crossing gadgets of size (at most) $11 \times 11$;
    directed edges model input and output wires.
    }
    \label{figure:boxGadgets}
\end{figure}

Finally, consider the layout of a \emph{connector gadget} depicted in \autoref{subfigure:connectorGadget}.
The gadget is used to systematically connect an instruction gadget to either an instruction gadget of the same \geomCell gadget or an instruction gadget of a neighbouring%
\footnote{When connecting vertically neighboured \geomCell gadgets, the connector gadget of \autoref{subfigure:connectorGadget} is rotated by $90^\circ$.}
\geomCell gadget.
In a first step, a connector gadget is placed between two neighbouring \geomCell gadgets as shown in \autoref{figure:connectorGadgetExample} for the horizontal case.
Then, for each instruction $I_i$ of $M$ we proceed as follows:
If $I_i = (i, r, j)$ is an instruction that adds one to the first register, we extend the outgoing wire $r_i$ straight to the right, pass the crossing gadgets of the connector gadget until we meet the wire $l_j$ running towards the left input wire of the instruction gadget $j$ of the neighbouring \geomCell.
At this position, we remove the crossing gadget and attach wire to wire.
Similarly, we proceed for additions to the second register and subtraction instructions.
\begin{figure}[bt]
    \centering
    \begin{subfigure}[b]{0.26 \textwidth}
        \centering
        \includegraphics[height=3.2cm, page=6]{./figures/cellGadget}
        \caption{A \geomCell gadget.}
    \end{subfigure}
    \begin{subfigure}[b]{0.44 \textwidth}
        \centering
        \includegraphics[height=3.2cm, page=1]{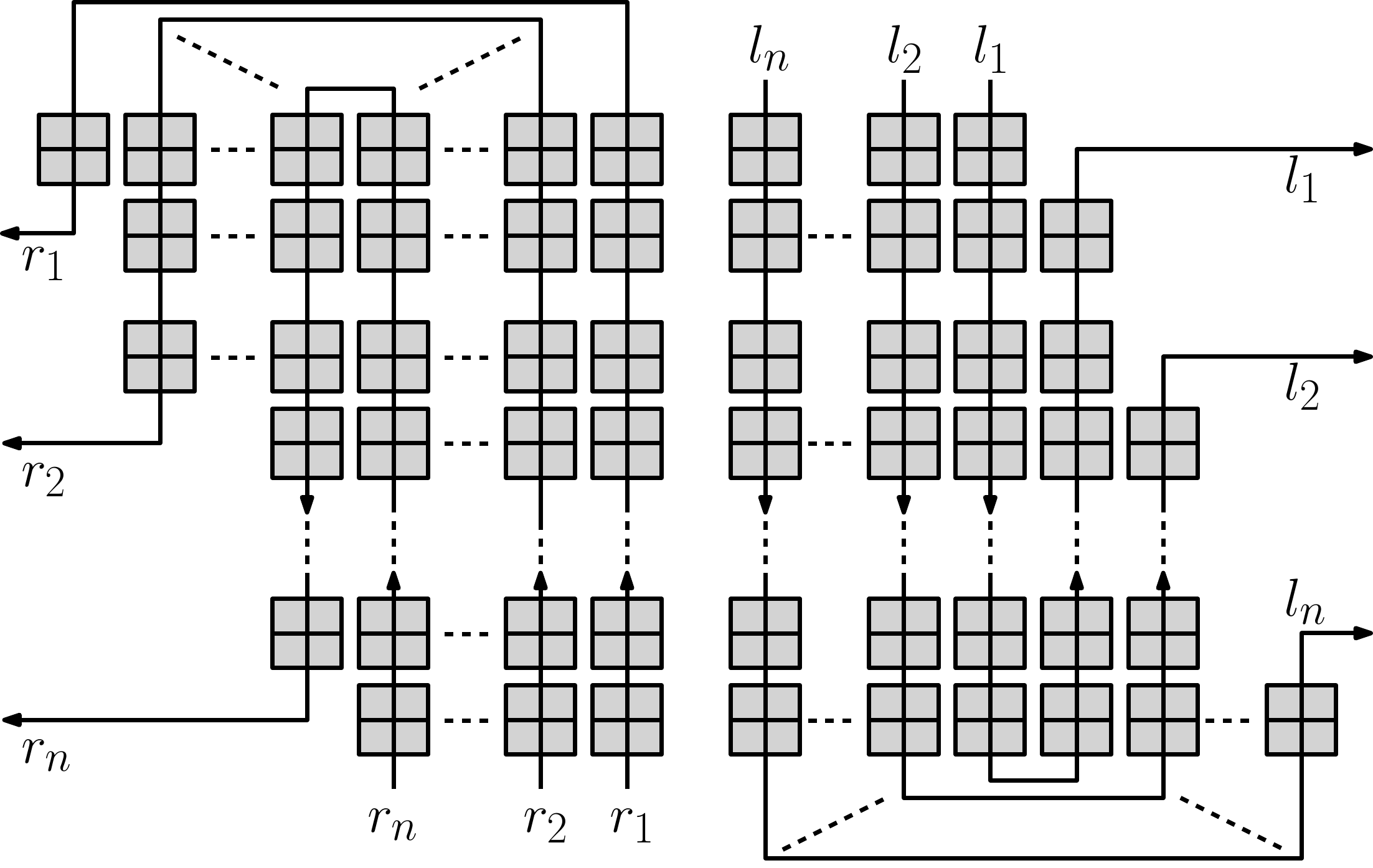}
        \caption{A connector gadget.}
        \label{subfigure:connectorGadget}
    \end{subfigure}
    \begin{subfigure}[b]{0.26 \textwidth}
        \centering
        \includegraphics[height=3.2cm, page=5]{./figures/cellGadget}
        \caption{Neighbour.}
    \end{subfigure}
    \caption{The connector gadget is placed between two neighbouring \geomCell gadgets.}
    \label{figure:connectorGadgetExample}
\end{figure}

%

Let $(d_1, d_2, I_i)$ be the starting configuration of a two-register machine $M$ and $c$ be the central cell of the instruction gadget corresponding to $v_i$ in the \geomCell with coordinates $(d_1, d_2)$.
As described in the geometric interpretation above, a sequence of configurations of $M$ corresponds to a unique directed path in the graph $G_M$ and vice versa.
Similarly, the vertices of this path in $G_M$ correspond to a unique sequence of burning instruction gadgets when setting $c$ on fire.

Minsky \cite{article:minsky1961recursive} proved that for any Turing machine $T$, there is a two-register machine program $M$ that computes the same function.
His proof implies that:
$T$ halts with an empty semi-infinite tape and the head placed right of the square which indicates the end of the tape if and only if $M$ halts with the register contents $r_1 = 3$ and $r_2 = 0$.
Due to the construction of our instance of the firefighting model induced by $M$, this holds if and only if the fire propagation eventually ignites the central cell of the instruction gadget corresponding to $v_n$ of the \geomCell with coordinates $(3,0)$.
Consequently, any algorithm deciding the problem of \autoref{theorem:undecidability} can be used to decide the following special case of the general halting problem:
Given a Turing machine $T$ with a single, semi-infinite tape and an input $w\in\lbrace 0,1\rbrace^*$; 
does $T$ halt with an empty tape and the head placed right of the tapes end, when started with input $w$?
The undecidability of this problem implies the undecidability of the decision problem of \autoref{theorem:undecidability} and completes the proof.

\FloatBarrier
\subsection{Omitted proof from \autoref{subsection:villageProtection}}
\label{appendix:intersectionLemma}
\setcounter{lemma}{1}
\begin{lemma}
	\label{lem:intersection}
	Any shortest local-cost $s$-$t$-path with \twist $1$ is simple.
\end{lemma}

\begin{proof}
	Assume $\wall_b$ is a shortest local-cost $s$-$t$-path with \twist $1$ and at least one intersection.
	Then we will construct a cheaper path $\wall_b'$ with no intersections to prove a contradiction.
	
	The path $\wall_b$ can be represented as a list of $n+2$ vertices $s, v_1, v_2, \ldots, v_i$, $\ldots, v_{n-1}, v_n, t$.
	It contains intersections as long as there exists a reappearing vertex, so $v_i = v_j$ for some distinct $i \neq j$.
	Let $v = v_i$ be the first such vertex on the path and $v_j = v$ its next appearance.
	Then a new path can be constructed by removing everything between $v_i$ and $v_{j+1}$. 
	Repeating this until no vertex appears twice in the path results in a path free of intersections.
	
	However, while removed edges can no longer contribute to the cost of the path, removing them can change the type and hence the local cost of the edges from $v$ up to the first unaffected $r$-edge after $(v,v_{j+1})$. Therefore we must look at the specific situation at $v$ more closely to analyse the change in cost.
	
	The situation at the intersection can take one of the following four cases illustrated in \autoref{fig:intersectionCases}:
	\begin{enumerate}
		\item $\wall_b$ turns right at $v_i$ and left at $v_j$.
		\item $\wall_b$ turns left at $v_i$ and left at $v_j$.
		\item $\wall_b$ turns right at $v_i$ and right at $v_j$.
		\item $\wall_b$ turns left at $v_i$ and right at $v_j$.
	\end{enumerate}
	Notice that these are all cases because the path cannot include $v_{i-1}$ a second time as we assumed $v_i$ to be the first reappearing vertex.
	
	\begin{figure}
		\centering
		\includegraphics[width=\textwidth]{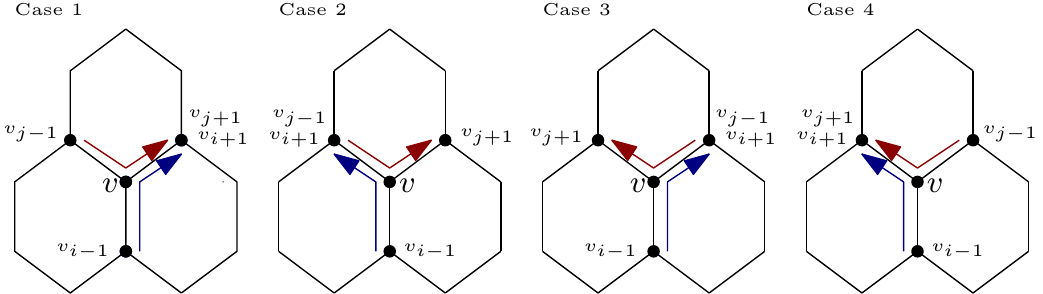}
		\caption{$\wall_b$ follows the blue edges upon its first visit of $v$ and the red edges on its second visit. Cases are equivalent for rotation.}
		\label{fig:intersectionCases}
	\end{figure}
	
	Recall the local-cost function:
	edges of type $r$ and $l_1$ cost between $0$ and $\burningTime$; edges of type $l_2, l_3$, and $l_4$ always cost $\burningTime$.
	Also for edges along the same cell, type $r$ edges are at most as costly as $l_1$ edges and both are at most as costly as edges of type $l_2, l_3$, and $l_4$.
	Hence, edges can only become more expensive if they were of type $r$ or $l_1$ on $\wall_b$ and the new cost can be at most $y$.
	
	In cases $1$ and $2$, $(v,v_{j+})$ is an edge of type $l_1$ (or $l_2, l_3$ or $l_4$) on $\pi_b$ and all further edges until the next $r$ edge are of type $l_2, l_3$, or $l_4$.
	On $\pi_b'$,  $(v,v_{j+})$ is of type $r$, so the local-cost of $\pi_b'$ can only be less or equal to that of $\pi_b$ in these cases.
	
	In case $3$ and $4$, $(v,v_{j+1})$ is an edge of type $r$ on $\wall_b$ and $(v_{j+1},v_{j+2})$ could be of type $l_1$, so they could increase in cost by $y$ each.
	All further edges until the next $r$ edge are of type $l_2, l_3$, or $l_4$ and hence do not matter.
	So $\wall_b'$ might have local-cost higher than $\wall_b$ by $2y$, if the removed part of the path $v_i, v_{i+1}, \ldots, v_{j}$ had cost less than $2y$.
	But then, the removed part can contain at most $2$ edges of type $l_k$ with $k > 1$.
	But this is not enough to reach the edge $(v_{j-1}, v)$ from $(v, v_{i+1})$ without taking more right turns than left turns.
	So in both these cases, the removal of that piece of the path either does not increase the cost of $\wall_b'$ or the \twist of $\wall_b'$ is at least $6$ higher than that of $\wall_b$.
	
	
	But at the end of the procedure we definitely arrive at a simple path, and hence at a path $\wall_b'$ with \twist $1$.
	Hence at some later intersection we will remove a part of the path, such that the \twist of $\wall_b'$ decreases by $6$.
	Considering that removing a part changes the type of at most one of the edges still in $\wall_b'$ to $r$ in any of the cases given above, the removed part contains at least $5$ more edges of type $l_1, l_2, l_3$, or $l_4$ than of type $r$.
	That means, it must include at least $5$ edges of type $l_2, l_3$ or $l_4$, which means the cost of $\wall_b'$ is at least $5y$ lower than that of $\wall_b$ which counters the cost of this removal.
	In fact similar arguments can be made about the removed parts of the path in case $1$ and $2$, so that in total, the cost of $\wall_b'$ is strictly less than $\wall_b$.
\end{proof}

\subsection{Selective fortification algorithm}
\label{appendix:IndividualAlgorithm}
The algorithm computes a solution for the selective fortification problem by computing shortest local cost separating paths from any vertex $s$ at the bottom boundary of a rectangular grid to any vertex $t$ at the top boundary.
$\LocalCost{(v_1,v_2), e}$ therein always denotes the local cost $c$ of an edge $(v_1, v_2)$ of type $e$.

Let $e$ be an edge along a cell $c$ of resistance $\ignitionCounter_e$. Remember that we also assume, that \burningTime is identical for all cells.
We say $e$ is of type $r$ if it comes after a right turn or is the very first edge of the path.
This will mean that it goes along a different left-hand cell than the previous edge on the path.
We say $e$ is of type $l_k$ if it comes after the $k^{th}$ consecutive left turn.
This will make it the $k+1^{th}$ consecutive edge along the same left-hand cell.
Then, the local cost $c$ of $e$ depends on $\ignitionCounter_e$ and its type:

\begin{equation*}
c(e, \text{type}) =
\begin{cases}
\max (0, y+1-\ignitionCounter_e), & \text{if type} = r\\
\min (y, 2y+1-\ignitionCounter_e), & \text{if type} = l_1\\
y, & \text{if type} = l_2,l_3,l_4\\
\end{cases} 
\end{equation*}

The algorithm makes use of a priority queue $\queue$ storing tuples $(v, p, e, w)$ storing the shortest known local cost path to the vertex $v$ via a predecessor $p$, where the last edge $(p,v)$ is of edge type $e$ and the path has twist $w$.
It functions similar to Dijkstra's shortest path algorithm, but it needs the additional information $p$ and $e$ in the tuple to correctly calculate the cost of the next edge via the function $\LocalCost{}$.
In addition, it separates paths by the twist $w$ in the tuples in the $\queue$ to filter out an intersection free path in the end.
As this introduces potentially infinite tuples, we use the \twist limit based on the size of \rectangle to block tuples with too high or low twist from ever entering \queue.

The correctness of the algorithm follows from the correctness of Dijkstra's shortest path algorithm in combination with arguments from \autoref{subsection:villageProtection}.
While this algorithm only calculates the cost a solution for the selective fortification problem, it can easily be adjusted to output the whole solution path.

\begin{procedure}[hbt]
	\caption{Selective fortification path($\mathcal{R}$)}
	\label{algorithm:selectiveAlgorithm}
	\Input{The dual graph of a rectangular grid \rectangle in our basic hexagonal model, and the upper twist limit $w_{max}$ based on its width and height.}
	\Output{The cost of an optimal solution for the selective fortification problem on \rectangle}
	\BlankLine
	\ForEach{cell $s$ on the bottom boundary of \rectangle and the vertex $v_s$ directly above it}{
		\Update{\LocalCost{(s, v_s), r}, v_s, s, r, 1}\\
	}
	\While{\queue is not empty}{
		$(c, v, p, e, w) \gets \ExtractMin{}$\\
		\If{$v$ is a cell $t$ on the top boundary of \rectangle}{
			\If{$w=1$}{
			\Return $c$\\}}
		\Else {
			$v_r \gets \text{vertex reached when taking a right turn after the edge $(p,v)$}$\\
			\If{$w < w_{max}$}{
				\Update{c + \LocalCost{(v, v_r), r}, v_r, v, r, w+1}\\
			}
			$v_l \gets \text{vertex reached when taking a left turn after the edge $(p,v)$}$ \\
			\If{$w > -w_{max}$}{
				\Switch{$e$}{
					\Case{$r$}{
						\Update{c + \LocalCost{(v, v_l), l_1}, v_l, v, l_1, w-1}\\
					}\Case{$l_1$}{
						\Update{c + \LocalCost{(v, v_l), l_2}, v_l, v, l_2, w-1}\\
					}\Case{$l_2$}{
						\Update{c + \LocalCost{(v, v_l), l_3}, v_l, v, l_3, w-1}\\
					}\Case{$l_3$}{
						\Update{c + \LocalCost{(v, v_l), l_4}, v_l, v, l_4, w-1}\\
					}\Case{$l_4$}{
						\Update{c + \LocalCost{(v, v_l), l_5}, v_l, v, l_5, w-1}
					}
				}
			}
		}
	}
\end{procedure}

\end{document}